\documentclass[10pt,journal]{IEEEtran}
\usepackage{url}
\usepackage{float} 
\usepackage{pifont} 
\usepackage{amsmath,amsfonts,amssymb,mathrsfs,amsthm,amsbsy,graphicx,color,bm,algpseudocode,balance,mathtools,xfrac,nccmath}
\usepackage{subcaption}
\usepackage{cite}
\usepackage[font={small}]{caption}

\usepackage[table,xcdraw]{xcolor}
\newcommand{\vc}[3]{\overset{#2}{\underset{#3}{#1}}}

\newtheorem{proposition}{Proposition}
\ExplSyntaxOn
\NewDocumentCommand{\MeijerG}{smmmm}
{ \IfBooleanTF{#1}
{\vic_meijerg:nnnnnn { #2 } { #3 } { #4 } { #5 } { small } { }}
{\vic_meijerg:nnnnnn { #2 } { #3 } { #4 } { #5 } { } { \; }}}
\seq_new:N \l__vic_meijerg_args_in_seq
\seq_new:N \l__vic_meijerg_args_out_seq
\cs_new_protected:Nn \vic_meijerg:nnnnnn
{\seq_set_split:Nnn \l__vic_meijerg_args_in_seq { | } { #3 }
\seq_clear:N \l__vic_meijerg_args_out_seq  
\seq_map_inline:Nn \l__vic_meijerg_args_in_seq
{\seq_put_right:Nn \l__vic_meijerg_args_out_seq
{\begin{#5matrix} ##1 \end{#5matrix}}}
G\sp{#1}\sb{#2}
\left(\seq_use:Nn \l__vic_meijerg_args_out_seq { #6\middle|#6 }
#6\middle|#6
#4\right)}
\ExplSyntaxOff
\usepackage{soul}
\sethlcolor{yellow} 

\usepackage{tabularx}      
\usepackage[table]{xcolor}  
\usepackage{booktabs}       
\usepackage{array}         
\definecolor{LightGray}{gray}{0.9}
\definecolor{DarkGray}{gray}{0.3}

\begin{document}

\title{Improved Energy-Based Signal Detection for Ambient Backscatter Communications}

\author{S. Zargari,  C. Tellambura, \textit{Fellow, IEEE}, and A. Maaref,~\IEEEmembership{Senior Member,~IEEE}
\thanks{This work was supported in part by Huawei Technologies Canada Company, Ltd.}
\thanks{Shayan Zargari and Chintha Tellambura are with the Department of Electrical and Computer Engineering, University of Alberta, Edmonton, AB T6G 1H9, Canada (e-mail: zargari@ualberta.ca; ct4@ualberta.ca).}
\thanks{Amine Maaref is with the Ottawa Wireless Advanced System Competency Centre, Huawei Canada, Ottawa, ON K2K 3J1, Canada (e-mail: Amine.Maaref@huawei.com).} \vspace{-10mm}}

\maketitle

\begin{abstract}
In ambient backscatter communication (AmBC) systems, passive tags connect to a reader by reflecting an ambient radio frequency (RF) signal. However, the reader may not know the channel states and RF source parameters. The traditional energy detector (TED) appears to be an ideal solution. However, it performs poorly under these conditions. To address this, we propose two new detectors: (1) A joint correlation-energy detector (JCED) based on the first-order correlation of the received samples and (2) An improved energy detector (IED) based on the $p$-th norm of the received signal vector. We compare the performance of IED and TED under the generalized noise model using McLeish distribution and derive a general analytical formula for the area under the receiver operating characteristic (ROC) curves. Based on our results, both detectors outperform TED. For example, the probability of detection with a false alarm rate of $1\%$ for JCED and IED is $22.97\%$ and $5.41\%$ higher, respectively, compared to TED for a single-antenna reader. Using the direct interference cancellation (DIC) technique, these gains are $34.92\%$ and $3.7\%$, respectively. With a four-antenna array at the reader and a $5\%$ false alarm rate, the JCED shows a significant BER improvement of $28.68\%$ without DIC and $48.21\%$ with DIC.
\end{abstract}


\section{Introduction}

\IEEEPARstart{T}{he} 
Internet of Things (IoT) aims to connect billions of devices through various networks \cite{3GPP}. However, maintaining these devices, such as battery charging or replacement, becomes challenging as the device count grows exponentially \cite{Nguyen_Van}.  Hence, passive or battery-free IoT devices are gaining popularity, offering low-cost connectivity and high energy efficiency. In this context, ambient backscatter communication (AmBC) is a promising low-power wireless solution, enabling tags to communicate with a reader by reflecting ambient radio frequency (RF) sources like cellular base stations (BSs), television (TV) towers, and wireless fidelity (Wi-Fi) access points (APs) \cite{Fatemeh_Rezaei,Xu10195838}. These tags operate without batteries, relying on energy harvesting (EH).

AmBC offers advantages in deploying passive IoT networks by avoiding complex modulation/demodulation, noise amplification, and self-interference at the tag \cite{Zargari10353962,Nguyen_Van}. For instance, \cite{Yongjun10109835} presents strategies for optimizing resource allocation in wireless-powered backscatter communication systems, emphasizing the need for robust and efficient resource management in large-scale IoT deployments. However, signal detection remains a key technical challenge due to 1) severe dyadic pathloss and 2) channel estimation \cite{Zargari10353962,Nguyen_Van,Fatemeh_Rezaei}.

This paper aims to develop signal detectors to overcome these challenges. To set the context and background, we first begin by reviewing current coherent and non-coherent detectors (refer to Table \ref{tab:litSummary1}).

\begin{table*}[t]
\centering
\captionsetup{font=small,labelfont={color=DarkGray,bf},textfont={color=DarkGray}}
\caption{Summary of existing AmBC signal detectors. }
\begin{tabularx}{\textwidth}{
>{\raggedright\arraybackslash}p{0.9cm}
>{\raggedright\arraybackslash}p{3.5cm}
>{\raggedright\arraybackslash}p{2.2cm}
>{\raggedright\arraybackslash}X
>{\raggedright\arraybackslash}X
}
\toprule[1.5pt]
\rowcolor{LightGray}
\textbf{\color{DarkGray} Works} & \textbf{\color{DarkGray} Detection Methods} & \textbf{\color{DarkGray} Type} & \textbf{\color{DarkGray} Pros} & \textbf{\color{DarkGray} Cons} \\
\midrule[1pt]
\cite{Yiwen_Tao} & Nonlinear filtering & Coherent & Enhances SNR using simpler analog circuits, automatic decision threshold & Requires ambient RF signal and CSI \\ \midrule 
\cite{Jing_Qian_4} & ML detector & Coherent & Evaluates outage probability, closed-form expression for BER & Requires precise ambient RF signal parameters and CSI \\ \midrule 
\cite{Sudarshan_Guruacharya} & GLRT, Bayesian testing & Non-coherent & Bayesian detector offers better performance & Bayesian detector has higher computational complexity \\ \midrule 
\cite{Zeng_Tengchan} & Statistical covariance method & Non-coherent & Outperforms TED in low SNR regions & Limited to low SNR regions \\ \midrule 
\cite{Kartheek_Devineni} & Energy detector & Non-coherent & Analysis using noncentral chi-squared distribution, derives ergodic BER & Relies on specific statistical modeling \\ \midrule 
\cite{Devineni_Kartheek_2} & Mean threshold, approximate ML detectors & Non-coherent & Extends BER analysis, works under both perfect and imperfect CSI & Performance varies with CSI accuracy \\ \midrule 
\cite{Yang9950349} & Parallel detection & Non-coherent & Does not require CSI & Handling interference and tag waveforms might become complex in dense IoT settings.
\\ \midrule 
{{\cite{Yunkai_Hu_2,Xiyu_Wang,Qianqian_Zhang_1}}} & {KNN, SVM, GMM}  &{ Machine learning} & {Demonstrates adaptability by effectively handling varying signal conditions } &  {Requires CSI with high computational complexity; sensitive to noise and label pollution; depends on initial label transmission, adding overhead.} 
\\ \midrule 
This work & JCED and IED  & Non-coherent & Improved BER performance, analyzing under different noise conditions & Complexity in optimizing $p$ for IED, requires optimization for JCED \\
\bottomrule[1.5pt]
\end{tabularx}
\label{tab:litSummary1}
\end{table*}

\subsection{Coherent detection in AmBC systems}
In principle, coherent detection requires precise signal phase information. The receiver thus extracts data from weak signals amidst noise by exploiting phase relationships with a reference signal.  This approach offers high flexibility in the choice of modulation formats.  However, its implementation is challenging due to the requirement of precise knowledge of ambient RF signal parameters and channel state information (CSI).

To address this, the reader's signal-to-noise ratio (SNR) can be improved without using precise knowledge of ambient RF signals or CSI. In \cite{Yiwen_Tao}, a nonlinear filtering method is proposed to enhance SNR using simpler analog circuits, with an automatically calculated decision threshold for data recovery from the sampled signal. Reference \cite{Jing_Qian_4} analyzes symbol detection and evaluates the outage probability as a function of system power using the maximum likelihood (ML) detector, incorporating a closed-form expression for bit error rate (BER).

\subsection{Non-coherent detection in AmBC systems} 
This approach eliminates the need for carrier phase or CSI, reducing receiver complexity. However, it may lead to lower spectral efficiency and degraded performance.

In \cite{Sudarshan_Guruacharya}, the joint probability density function (PDF) of the received signal is analyzed to develop two optimal non-coherent detectors for AmBC, based on generalized likelihood ratio test (GLRT) and   Bayesian testing. The authors in \cite{Zeng_Tengchan} propose a statistical covariance method that outperforms the TED in low SNR regions. It provides two statistics based on the probability of detection and the BER. Reference \cite{Kartheek_Devineni} demonstrates that the exact conditional PDF of the average signal energy follows a noncentral chi-squared distribution. Further, \cite{Devineni_Kartheek_2} extends the BER analysis proposed in \cite{Kartheek_Devineni} by investigating three different detectors under both perfect and imperfect CSI assumptions.

Next, \cite{Yang9950349} introduces a parallel detection approach for AmBC systems, enabling simultaneous detection of signals from multiple tags. The method uses a modified expectation maximization algorithm to cluster received signals, allowing for parallel detection of multiple tags. This clustering is enhanced by the state-flipping characteristic of tag waveforms, which aids in differentiating between signals from different tags. Efficiency is achieved by eliminating the need for CSI and reducing complexity and overhead.

{In addition, references {\cite{Yunkai_Hu_2,Xiyu_Wang,Qianqian_Zhang_1}} further investigate machine learning algorithms such as k-nearest neighbors (KNN), support vector machine (SVM), and Gaussian mixture models (GMM) for  AmBC signal detection.} 

\subsection{Motivation and Contributions}
We propose two cost-effective, non-coherent AmBC detectors to overcome the challenges faced by AmBC systems. However, it is important to highlight potential use cases for AmBC systems. Industries like logistics, warehousing, and manufacturing rely on accurate and up-to-date databases for goods, materials, and assets. Tag-based tracking can enhance productivity, provide real-time data, and improve quality control. AmBC systems are designed to operate passively, leveraging the existing RF signals in the environment for both communication and EH. This dual-purpose utilization of ambient RF signals enables the operation of devices even in low RF energy density environments. Moreover, as the tags operate passively and rely on EH capabilities, AmBC systems serve as a more appropriate foundation for passive IoT systems, which is also being explored by the 3rd generation partnership project (3GPP) \cite{3gpp_RWS}.

Our main technical  contributions can be summarized as follows:
\begin{itemize}

\item  We propose a joint correlation-energy detector (JCED), which involves a test statistic combining the energy and the correlation of received samples \cite{Morteza_Tavana}. This is a new idea that has not been previously explored in the literature \cite{Yiwen_Tao,Jing_Qian_6,Anran_Wang,Pengyu_Zhang, Sudarshan_Guruacharya,Jing_Qian_4,Zeng_Tengchan,Kartheek_Devineni,Devineni_Kartheek_2,Kartheek_Devineni}. This linear combination of hypothesis tests leads to a greater mean distance between the distributions for both hypotheses, resulting in improved BER performance. Finally, the false alarm and detection probabilities are derived. 

\item This paper also proposes an improved energy detector (IED), also known as the $p$-norm detector, where $p$ is a tunable parameter. Many previous works\cite{Sudarshan_Guruacharya, Jing_Qian_4,Zeng_Tengchan,Kartheek_Devineni,Devineni_Kartheek_2}  can be viewed as a special case of our detector with fixed $p=2$, whereas the IED uses an arbitrary positive power operation with $p>0$. Generally, the optimal value of $p$ depends on detection and false alarm probabilities and the SNR. The traditional energy detector (TED) is given by $p=2$. The paper also provides a closed-form analysis of the probabilities of detection and false alarm for sufficiently large samples.

\item  In many practical scenarios such as urban environments or IoT networks, noise can exhibit dependence structures or follow non-Gaussian distributions due to multipath propagation, interference from other devices, fading effects, or other irregularities \cite{Xiaomei,Nikolaos}. To account for this, we also consider the McLeish distribution as a generalized model for noise channels. It is a robust model, capable of representing both Gaussian and non-Gaussian noise channels \cite{mcleish1982robust}, and it shares many of the characteristics of the Gaussian distribution, including its unimodality, symmetry, finite moments, and heavy tails \cite{Yilmaz}. With this model for a generalized noise, we derive closed-form expressions for the probabilities of detection and false alarm. 

\item We initially focus on a single-antenna model to thoroughly analyze the functionalities and performance intricacies of the JCED and IED.  To leverage the advancements in AmBC, we also extend our study to incorporate a multi-antenna reader system.

\item Numerical results demonstrate that the IED outperforms the TED through optimized power operation. The superiority of the IED in the presence of Laplacian noise is also highlighted compared to other EDs. The proposed JCED surpasses both the IED and TED using the correlation between adjacent signals. A case study employing the direct interference cancellation (DIC) technique yields valuable insights, and a general analytical formula for the area under the curve (AUC) of the receiver operating characteristic (ROC) is derived for the proposed detectors.

\end{itemize}

{\it Notations}: Vectors and matrices are denoted by boldface lowercase and uppercase letters, respectively, with $\mathbf{A}^T$ representing the transpose. The absolute value and real part of $z$ are indicated by $|z|$ and $\Re\{z\}$, respectively. A circularly symmetric complex Gaussian (CSCG) random vector with mean $\boldsymbol{\mu}$ and covariance matrix $\mathbf{C}$ is expressed as $\sim \mathcal{C}\mathcal{N}(\boldsymbol{\mu},\mathbf{C})$. The McLeish distribution for a complex, circularly symmetric random variable is denoted by $\mathcal{CML}(\mu, \sigma^2, q)$, characterized by mean $\mu$, variance $\sigma^2$, and non-Gaussianity parameter $q$. The notation $\partial_{x}$ signifies the partial derivative over  $x$. Expectation and variance are represented by $\mathbb{E}[\cdot]$ and $\text{Var}[\cdot]$, respectively. $\mathbb{C}^{N\times 1}$ and $\mathbb{R}^{N \times 1}$ denote complex and real vectors of dimension $N\times 1$. The PDF of a random variable $x$ is given by $f_x(\cdot)$. The Gaussian Q-function and its inverse are denoted by $\mathcal{Q}(\cdot)$ and $\mathcal{Q}^{-1}(\cdot)$. The Gamma function is represented by $\Gamma(\cdot)$ \cite[Eq. (8.310.1)]{gradshteyn2014table}, with $\Gamma(\cdot, \cdot)$ for the upper incomplete Gamma function \cite[Eq. (8.350.2)]{gradshteyn2014table}. The $q$-th order modified Bessel function of the second kind is $K_q(\cdot)$ \cite[Eq. (8.432)]{gradshteyn2014table}, and the Meijer’s G-function is expressed as $\MeijerG{m,n}{p,q}{\cdot }{\cdot}$ \cite[Eq. (9.301)]{gradshteyn2014table}.

\begin{figure}
\centering
\includegraphics[width=3in]{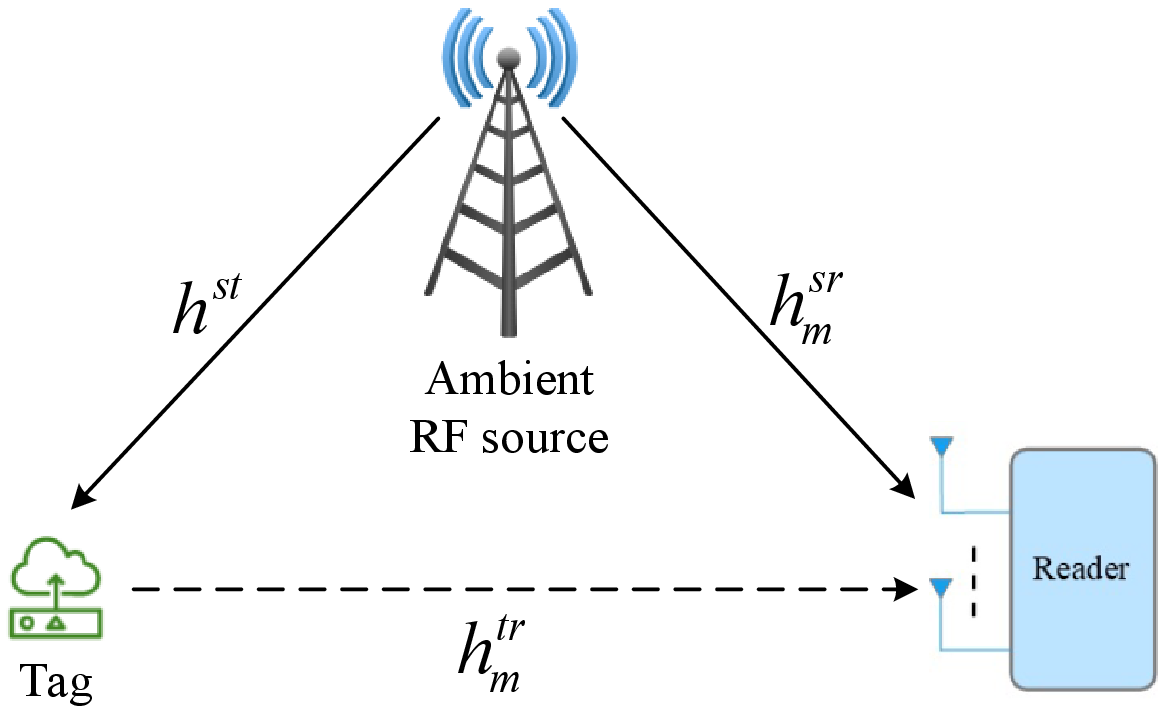}
\caption{AmBC network of an RF signal source, a reader, and a passive tag.  \vspace{-2mm}} \label{system_model} 
\end{figure}

\section{Background}
The parameters and statistics of the received signal are needed to determine an optimal detector. However, since these are generally unavailable, studies  \cite{Jing_Qian_4, Jing_Qian_2, Zeng_Tengchan} have investigated the TED. In particular, the TED test statistic is based on the average energy of the received signal samples \cite{Urkowitz}. Let $y(n)$, $n=\{1,2,\ldots N\}$, be the received signal at the reader. Then, the reader determines the data by averaging the received signal energy over $N$ samples given as
\begin{equation}\label{E_simple}
E_1 =\frac{1}{N}\sum^N_{n=1} |y(n)|^2.
\end{equation}
The TED statistic is used for binary hypothesis testing. We first describe that testing before moving on to the signal model. 
The following are the basic facts of binary hypothesis testing  \cite{levy2008principles}: 
\begin{enumerate}
\item[$\bullet$]  Hypotheses and a-priori probabilities: The hypotheses are  $\mathcal{H}_0$ and $\mathcal{H}_1$ which denote the null hypothesis and the signal hypothesis, respectively. Their corresponding a-priori probabilities are given by $\pi_0=\Pr(\mathcal{H}_0)$ and $\pi_1=\Pr(\mathcal{H}_1)$, where $\pi_0+\pi_1=1$.

\item[$\bullet$]  Observation: Consider a random vector $\mathbf{Y}$ with sample space  $\mathcal{Y}$, then, an observation is a sample vector $\mathbf{y}\in \mathbb{C}^{N\times 1}$ of $\mathbf{Y}$.

\item[$\bullet$]  Hypotheses: The conditional probabilities are used to represent the connection  between hypotheses and observation as follows:

\begin{equation}
\begin{array}{l}
\mathcal{H}_0:  ~ \mathbf{Y}	\sim f_{\mathbf{Y}}(\mathbf{y}|\mathcal{H}_0),\\ \mathcal{H}_1: ~ \mathbf{Y}	\sim f_{\mathbf{Y}}(\mathbf{y}|\mathcal{H}_1).
\end{array}  
\end{equation}

\item[$\bullet$] Decision function: Based on observation, the detector determines which hypothesis $\mathcal{H}_0$ or $\mathcal{H}_1$ is true. A decision function can be expressed as
\begin{equation}
\delta(\mathbf{y}) = \left\{ \begin{array}{l}
1 \quad \text{if decide on} \: \: \mathcal{H}_0 ,\\
0  \quad \text{if decide on} \: \: \mathcal{H}_1.
\end{array} \right.
\end{equation}

\item[$\bullet$]  Performance measure: The performance of detection can be measured based on  the probabilities of false alarm and detection, given as follows:
\begin{align}
&P_F(\delta)=\Pr(\mathcal{Y}_1 | \mathcal{H}_0)= \int_{\mathcal{Y}_1} f_{\mathbf{Y}}(\mathbf{y}|\mathcal{H}_0) d\mathbf{y}\label{P_FF},\\
&P_D(\delta)=\Pr(\mathcal{Y}_1 | \mathcal{H}_1)= \int_{\mathcal{Y}_1} f_{\mathbf{Y}}(\mathbf{y}|\mathcal{H}_1) d\mathbf{y},\label{P_DD}
\end{align}
respectively. This method chooses the test $\delta$ that maximizes $P_D(\delta)$ while guaranteeing that the probability of a false alarm, $P_F(\delta)$, is no more than $\eta$, which is called the size of the test. Since maximizing $P_D(\delta)$ and minimizing $P_F(\delta)$ cannot be achieved simultaneously, one can fix $P_F(\delta)$ and maximize $P_D(\delta)$, leading to the following optimization problem:
\begin{align}
\delta_{\text{NP}}=  \:\: &\underset{\delta}{\text{argmax} }   \: P_D(\delta),\quad \text{s.t.} \hspace{1em} P_F(\delta)\leq \eta. 
\end{align}
The Neyman-Pearson (NP) method evaluates decision rules using the ROC curve, which plots detection probability $(P_D)$ against false alarm probability $(P_F)$. A better-than-random detector's ROC lies above the diagonal, while a perfect one passes through $(0,1)$. The ROC curve is concave and satisfies $P_D \geq P_F$. The feasible region centers at $(0.5, 0.5)$, and the ROC slope at $(P_D(\tau), P_F(\tau))$ indicates the threshold $\tau$. The AUC summarizes classifier performance: a perfect classifier has AUC $1$, while a random one has $0.5$ \cite{levy2008principles}.

\end{enumerate}

\section{ AmBC System Model}\label{General_AmBC Systemodel}
Fig. \ref{system_model} depicts an AmBC network of an RF source, a $M$-antenna reader, and a passive tag. In the first part, all the nodes are assumed to be single-antenna devices. The tag encompasses both a modulation block and an EH block. For a detailed discussion of these blocks, we refer the reader to   \cite{Nguyen_Van,Azar_Hakimi_1,Fatemeh_Rezaei,Fatemeh10078244}.

\begin{enumerate}
\item \textbf{Backscatter modulation}: A tag, as a passive device, performs OOK modulation by switching between two load impedances to represent “0" or “1" by matching or mismatching the antenna impedance, respectively. The reflection coefficient of the tag is given by $\Gamma_{i} = \frac{Z_i-Z_{a}^\star}{Z_i+Z_{a}}$, where $Z_{a}$ denotes the antenna impedance of the tag and $Z_i$ is the load impedance of state $i = \{1, 2\}$  \cite{Nguyen_Van}. Specifically, $ \Gamma_i = |\Gamma_i| b_i$ where the tag can use distinct  $b_i$ values to send its data. In addition, the reflection coefficients of impedance values have a constant magnitude, i.e., $|\Gamma_i|^2=|\Gamma|^2=\xi \in (0, 1]$. In particular, $\xi = \vert \Gamma \vert^2$ denotes the power reflection coefficient at the tag satisfying $0 \leq \xi \leq 1$.  

\item \textbf{Synchronization}:  Although the tag must synchronize with the incoming signal, minor misalignment can still lead to only small errors in the backscattered signal, primarily because ambient RF sources typically transmit at significantly higher rates than backscatter tags. As demonstrated in \cite{Manideep}, the tag performs code-word translation on ambient 802.11b packets, embedding its data by altering the phase of each symbol in the packet. Despite potential minor synchronization errors, the backscattered signal still retains the 11-bit barker code and can be decoded by a reader. 
\end{enumerate}

\subsection{Channel models}
We use $ h = L_{\text{path}}^{1/2} g $ where small-scale fading coefficient  $ g \sim \mathcal{CN}(0,1) $ and $ L_{\text{path}} $ denotes macro fading effects. The model $ g \sim \mathcal{CN}(0,1) $ is predicated upon a propagation environment characterized by rich scattering that lacks direct line-of-sight (LoS) paths among the nodes \cite{Jing_Qian_4, Sudarshan_Guruacharya}. This modeling approach suits environments such as factories with obstructions and moving objects.  Similarly, tags enclosed in packaging or inside storage facilities may experience multipath propagation in storage or transportation scenarios.

A block fading model for all the channels can be assumed with a given coherence time\footnote{Although the reader might not possess full CSI, it retains the capability to approximate channel conditions using a conventional estimator \cite{Georgios_Vougioukas_1}. This estimation process involves the transmission of pilot signals, which are then reflected by the tags.}. 
The channel between the RF source-to-$m$-th antenna of the reader, the tag-to-$m$-th antenna of the reader, and the RF source-to-tag are assumed as zero-mean CSCG random variables as follows: $h^{sr}_m\sim\mathcal{CN}(0, \sigma^2_{sr})$, $h^{tr}_m\sim\mathcal{CN}(0, \sigma^2_{tr})$, and $h^{st}\sim\mathcal{CN}(0, \sigma^2_{st})$, respectively.

\subsection{Signal model}
We assume the random RF  signal from the  RF source is represented by  $s(n)$ with the transmit power of $P_s$. In particular, TV towers transmit up to $100$ kW for high frequency (VHF) and $1,000$ kW for ultra-high frequency (UHF) \cite{ecfr2023}. In 5G, 3GPP standards dictate transmit power: macro BSs have $43$ to $46$ dBm, while micro BSs range from $30$ to $36$ dBm \cite{3gpp-ts-38.104}.	Also, $s(n)$ and $s(\hat{n})$ are uncorrelated for $ n\neq \hat{n}.$ The baseband signal received by the tag at the $n$-th sampling instance can be expressed as follows:
\begin{equation}
x(n) = h^{st}s(n).
\end{equation}
We denote the backscatter modulation symbol of the tag as $b_k\in\{0, 1\}$. Notably, ambient RF sources transmit at rates significantly higher than the tags, making it reasonable to assume that $b_k$  remains consistent throughout $N$ observations. Thus, the signal backscattered by the tag can be represented as follows:
\begin{equation}
x_b(n) = \xi b_k x(n),  
\end{equation}
where $\xi \in (0, 1]$ is a scaling factor associated with the antenna gain and scattering efficiency of the tag \cite{Azar_Hakimi, Nguyen_Van}. At the $m$-th antenna of the reader, the signal corresponding to the tag symbol is given by 
\begin{align}
y_m(n) &\!\!\!=\! \left(h^{sr}_m  \!+\! \xi b_k h^{st}_m h^{tr}_m\right)s(n)   \!+\! w_m(n),~m=1,\cdots,M,
\end{align}
where $w_m(n)\sim\mathcal{CN}(0, \sigma^2_{w})$ represents AWGN and the noise samples are assumed to be independent. Further, the signal received by the $m$-th antenna of the reader can be described differently under various hypotheses as below:
\begin{equation}\label{received_signal_2}
y_m(n) = \left\{ \begin{array}{l}
h^{0}_m s(n) + w_m(n),~ \text{if} \:\:\mathcal{H}_0,\\
h^{1}_m s(n) + w_m(n), ~ \text{if} \:\:\mathcal{H}_1,
\end{array} \right.
\end{equation}
where $h^{0}_m = h^{sr}_m$ denotes the direct channel link (RF source-to-reader) and $h^{1}_m = h^{sr}_m + \xi  h^{st} h^{tr}_m$ indicates the dyadic channel link (RF source-to-tag-to reader).

\subsection{Direct interference cancellation}
Direct link interference (DLI) cancellation is possible when the reader can eliminate direct interference by decoding the symbols transmitted by the RF source. This may be possible in (a) orthogonal frequency division multiplexing (OFDM)-AmBC or (b) symbiotic radio.  For example, in \cite{Yang7841620}, the tag reflects OFDM symbols to send its data. The reader then cancels the direct interference by exploiting the unused part of the cyclic prefix (CP)\footnote{{Our circuit model integrates signal reception and uses the CP structure in OFDM systems {\cite{Yang7841620}} or channel estimation in symbiotic systems {\cite{Ruizhe, YangZ2018}} to remove DLI. It then processes and analyzes signals for power and pattern recognition, combining these findings via digital signal processing (DSP) or microcontroller unit (MCU) for effective backscatter signal detection.}}. Another intriguing method is the symbiotic radio approach, where a cooperative receiver merges the reader with the primary user to decode both user and tag signals \cite{Ruizhe, YangZ2018}. Here, we assume a symbiotic radio approach where the user has access to CSI and can estimate the RF signal through a pilot signal, facilitating the effective cancellation of direct interference {\cite{Ruizhe, YangZ2018}}.

Because of these potential DIC methods, we examine the performance of the proposed detectors in scenarios where DIC is employed. In this case, the signal received by the $m$-th antenna of the reader can be described as follows:
\begin{equation}
y_m(n) = \left\{ \begin{array}{l}
h^{\text{RI}}_m s(n) + w_m(n),~~~~\quad\quad\quad\quad  \text{if} \:\:\mathcal{H}_0,\\
(h^{\text{RI}}_m + \xi h^{st} h^{tr}_m )s(n) + w_m(n), ~\text{if} \:\:\mathcal{H}_1,
\end{array} \right.
\end{equation} 
where $h^{\text{RI}}_m\sim\mathcal{CN}(0, \epsilon)$ denotes residual interference (RI) at the $m$-th antenna of the reader, which follows a Gaussian distribution with zero mean and variance $\epsilon$. The parameter $\epsilon$ indicates the strength of the RI and is confined to the range $0\leq \epsilon\leq 1$. Under the null hypothesis $\mathcal{H}_0$, both RI and noise are present at the reader. Under the alternative hypothesis $\mathcal{H}_1$, the composite channel link is identical to the reflected link, and the DLI is removed via the DIC technique, leaving only RI.

\section{Detector Design: Single-Antenna Reader}

In this section, we consider a simple single-antenna reader setup. We examine the JCED, an advanced detector incorporating sample energy and first-order correlation to enhance BER. Then, we introduce the IED and TED. Notably, some channels have non-Gaussian noise\cite{Xiaomei}, which motivates us to investigate the performance of detectors under generalized noise conditions.

\subsubsection{\textbf{JCED}}
The JCED  utilizes both the energy and signal correlation to determine the presence of the signal. The test statistic for this detector is generated by combining the energy of the samples with first-order correlation values of the received signals through a linear combination, which is given by
\begin{align}\label{Z}
T_1&=\alpha \sum^{N-1}_{n=0}|y(n)|^2+\beta\sum^{N-2}_{n=0} y(n+1)y^*(n) = \alpha Z_1 + \beta Z_2,
\end{align}
for some $\alpha, \beta\in(0,1)$ and $\alpha+\beta=1$. This joint approach allows for a more reliable detection decision. According to the central limit theorem (CLT) theorem \cite{gnedenko1954limit}, we can conclude that for large values of $N$, the statistic $Z_1$ is expected to have a normal distribution. The resulting mean and variance of $Z_1$ can be presented as follows: 
\begin{align}
\mathbb{E}(Z_1) &= \left\{ \begin{array}{l}
(N+\gamma_0)\sigma_w^2,\quad \text{if} \:\:\mathcal{H}_0 ,\\
(N+\gamma_1)\sigma_w^2,\quad \text{if} \:\:\mathcal{H}_1,
\end{array} \right.\label{mean_U}\\[8pt] \nonumber
\text{Var}(Z_1) &= \left\{ \begin{array}{l}
(N+2\gamma_0)\sigma_w^4,\quad \text{if} \:\:\mathcal{H}_0 ,\\
(N+2\gamma_1)\sigma_w^4,\quad \text{if} \:\:\mathcal{H}_1.
\end{array} \right. 
\end{align}
where $\gamma_i=\frac{|h_i|^2E_s}{\sigma^2_w}, \forall i,$ and $E_s := \sum^{N-1}_{n=0}|s(n)|^2$. With minor modifications, the same argument can be made for $Z_2$. The terms $s(n+1)s^*(n)$ and $s(n+2)s^*(n+1)$ for this detector are not independent of one another. Therefore, we consider them two sets of independent and identically distributed (i.i.d) random variables.  Accordingly, $Z_2$ has an approximately normal distribution with a mean and variance of
\begin{align}
\mathbb{E}(Z_2) &= \left\{ \begin{array}{l}
|h_0|^2R_{ss}(1),\quad \text{if} \:\:\mathcal{H}_0 ,\\
|h_1|^2R_{ss}(1),\quad \text{if} \:\:\mathcal{H}_1,
\end{array} \right.\label{mean_R} \\ \nonumber
\text{Var}(Z_2)& = \left\{ \begin{array}{l}
(N-1)\sigma_w^4+2E_s|h_0|^2\sigma_w^2,\quad \text{if} \:\:\mathcal{H}_0 ,\\
(N-1)\sigma_w^4+2E_s|h_1|^2\sigma_w^2,\quad \text{if} \:\:\mathcal{H}_1,
\end{array} \right. 
\end{align}  
where $R_{ss}(1)$ corresponds to the autocorrelation function of the transmit signal from the RF source, which is defined as $R_{ss}(1) = \sum^{N-2}_{n=0}s(n+1)s^*(n)$. Also, it is assumed that the term $\sum^{N-2}_{n=0}|s(n+1)|^2+|s(n)|^2$ is approximated by $2E_s$ for large values of $N$.

\begin{figure}
\centering
\includegraphics[width=3.5in]{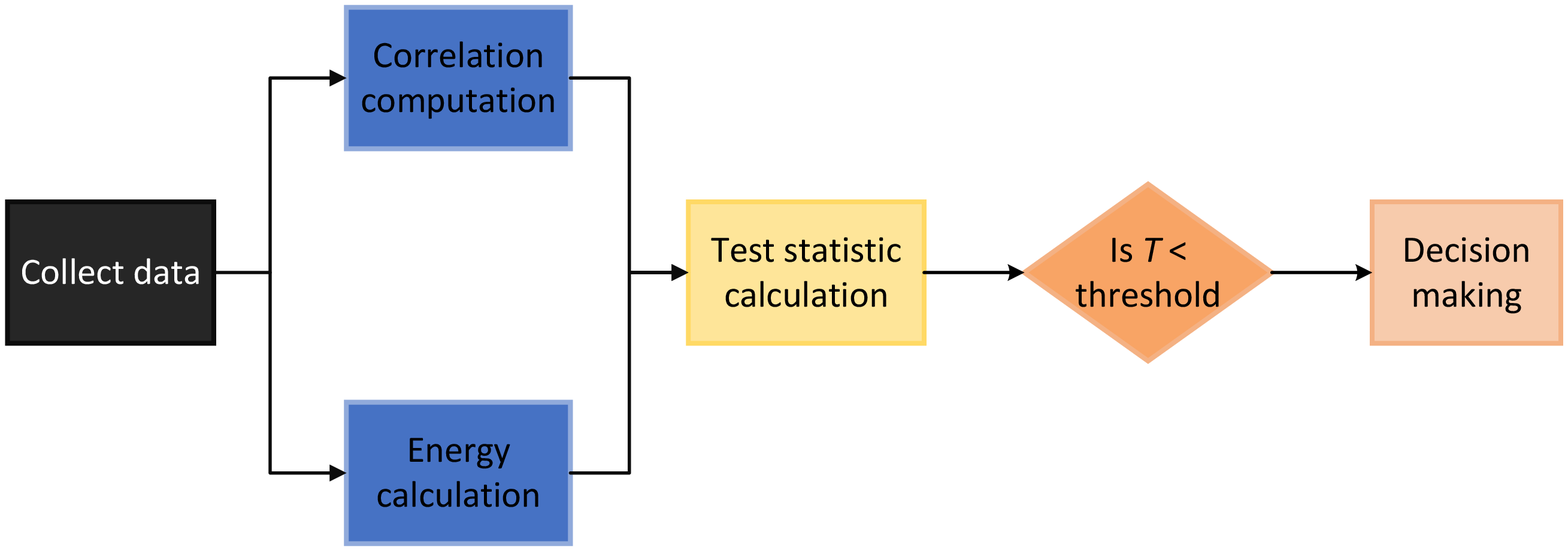}
\caption{The  JCED  binary hypothesis testing process.} \label{Diagram_jecd}
\end{figure}

\begin{proposition}\label{propro_2}
The covariance between $Z_1$ and $Z_2$ can be written as
\begin{equation}\label{covarince_correlatior}
\text{Cov}(Z_1, Z_2) = \left\{ \begin{array}{l}
2\sigma_w^2|h_0|^2 R_{ss}(1),\quad \text{if} \:\:\mathcal{H}_0,\\
2\sigma_w^2|h_1|^2 R_{ss}(1),\quad \text{if} \:\:\mathcal{H}_1. 
\end{array} \right.
\end{equation}
\end{proposition}
\begin{proof}
Please refer to Appendix \ref{Pro_1}.
\end{proof}
Using \eqref{mean_U} and \eqref{mean_R}, we can derive the mean and variance of $T_1$ under different hypotheses, as expressed below:
\begin{align}
\mathbb{E}(T_1) &= \left\{ \begin{array}{l}
\alpha (N+\gamma_0)\sigma_w^2+\beta|h_0|^2R_{ss}(1) ,\quad \text{if} \:\:\mathcal{H}_0 ,\\
\alpha (N+\gamma_1)\sigma_w^2+\beta|h_1|^2R_{ss}(1),\quad \text{if} \:\:\mathcal{H}_1,
\end{array} \right.\label{mean_C}\\
\text{Var}(T_1)&= \left\{ \begin{array}{l}
\mathbf{w}^T \Sigma_{\mathcal{H}_0} \mathbf{w},\quad \text{if} \:\:\mathcal{H}_0 ,\\
\mathbf{w}^T \Sigma_{\mathcal{H}_1} \mathbf{w},\quad \text{if} \:\:\mathcal{H}_1, 
\end{array} \right.\label{Var_Z}
\end{align}
where $\mathbf{w}= (\alpha,~\beta)^T$ includes the associate weights and 
\begin{align}
&\boldsymbol{\Sigma}_{\mathcal{H}_j}\!\! =\!\! \begin{bmatrix}
\!\!(N\!\!+\!\!2\gamma_j)\sigma_w^4 & 2|h_j|^2\sigma_w^2 R_{ss}(1)\\
2|h_j|^2\sigma_w^2  R_{ss}^*(1) & 2E_s|h_j|^2\sigma_w^2 \!\!+\!\!(N\!\!-\!\!1)\sigma_w^4 
\end{bmatrix}, \forall \! j \! \in \!\{0, 1\}.           
\end{align}
Accordingly, the decision rule can be stated as
\begin{equation}\label{39}
T_1 \vc{\gtreqless}{\mathcal{H}_1}{\mathcal{H}_0} \tau_1.
\end{equation}
Specifically, diagram \ref{Diagram_jecd} illustrates the binary hypothesis testing steps in the JCED framework, including initial signal analysis, correlation and energy computation, and concluding with decision-making based on predefined criteria and thresholds to determine signal presence or absence.   Based on \eqref{P_FF} and \eqref{P_DD}, we define the probabilities of false alarm and detection as follows:
\begin{align}
&P_F = \Pr(T_1>\tau_1|\mathcal{H}_0) =   Q\left(\frac{\tau_1-\mathbb{E}(T_1|\mathcal{H}_0)}{\sqrt{\text{Var}(T_1|\mathcal{H}_0)}}\right),\label{P_F}\\
&P_D = \Pr(T_1>\tau_1|\mathcal{H}_1)  =  Q\left(\frac{\tau_1-\mathbb{E}(T_1|\mathcal{H}_1)}{\sqrt{\text{Var}(T_1|\mathcal{H}_1)}}\right),\label{P_D}
\end{align}
respectively. Through some manipulation of \eqref{P_F} and \eqref{P_D}, the decision rule $\tau_1$ can be expressed as
\begin{equation}\label{42}
\tau_1=Q^{-1}(P_F)\sqrt{\text{Var}(T_1|\mathcal{H}_0)}+\mathbb{E}(T_1|\mathcal{H}_0).
\end{equation}
By substituting the expression given in \eqref{42} into \eqref{P_D}, we obtain the following result:
\begin{equation}
P_D=Q\left(\frac{\sqrt{\mathbf{w}^T\boldsymbol{\Sigma}_{\mathcal{H}_0}\mathbf{w}}Q^{-1}(P_F)+\mathbf{g}^T\mathbf{w}}{\sqrt{\mathbf{w}^T\boldsymbol{\Sigma}_{\mathcal{H}_1}\mathbf{w}}}  \right),
\end{equation}
where $\mathbf{g}=\left(\left(\gamma_0-\gamma_1\right)\sigma_w^2,~\left(|h_0|^2-|h_1|^2\right)R_{ss}(1)\right)^T$. Our goal is to maximize the probability of detection for a given probability of false alarm. In terms of the optimization problem, maximizing the probability of detection is equivalent to the following problem:
\begin{subequations}
\begin{align}\label{P2.1}
\text{(P2)}: & \hspace{1em} \max_{\mathbf{w}}  \hspace{1em}    	\frac{\sqrt{\mathbf{w}^T\boldsymbol{\Sigma}_{\mathcal{H}_0}\mathbf{w}}Q^{-1}(P_F)+\mathbf{g}^T\mathbf{w}}{\sqrt{\mathbf{w}^T\boldsymbol{\Sigma}_{\mathcal{H}_1}\mathbf{w}}}, \nonumber\\
&\quad\quad \text{s.t.} \hspace{1em} \alpha + \beta = 1.
\end{align} 
\end{subequations}
The decision threshold is proportional to the weight vector $\mathbf{w}$ for a given false alarm probability. However, problem (P2) is non-convex due to square roots and division in the objective function, making it challenging to solve. Methods like Fmincon’s interior point algorithm \cite{MatlabOTB} or the approach in \cite{Zhi_Quan} can address this nonlinearity and find a suboptimal solution within constraints.

\subsubsection{\textbf{IED}}
Since the classical  TED does not necessarily minimize the probability of false alarms or maximize the probability of detection as shown in \cite[Eq. (7.2)]{kay1993fundamentals},  we propose the use of  IED to increase the detection accuracy in AmBC systems. The test statistic and decision rule for the IED can be stated as follows:
\begin{equation}\label{improved_ED}
T_2 = \frac{1}{N}\sum^{N}_{n=1}\left(\frac{|y(n)|}{\sigma_w}\right)^p  \vc{\gtreqless}{\mathcal{H}_1}{\mathcal{H}_0}  \tau_2,
\end{equation}
where $p>0$ is an arbitrary constant and $\tau_2$ is the detection threshold to be determined. It is notable that \eqref{E_simple} and \eqref{improved_ED} only differ in the square operation. The detection threshold in \eqref{improved_ED} is adjusted according to the arbitrary positive power operation $p$. When $p=2$, the TED becomes a special case of the IED. Using \cite[Eq. 3.462.9]{gradshteyn2014table}, the mean and variance of $T_2$ under different hypotheses can be expressed as
\begin{align}
\mathbb{E}\left\{T_2|\mathcal{H}_i\right\}&  = \frac{2^{p/2}}{\sqrt{\pi}}\Gamma\left(\frac{p+1}{2}\right)\left(1+\gamma_i\right)^{p/2},\nonumber\\
\text{Var}\left\{T_2|\mathcal{H}_i\right\}& = \frac{2^{p}(1\!+\!\gamma_i)^p\Gamma(\frac{2p+1}{2})}{N\sqrt{\pi}}\!-\!\frac{2^p(1\!+\!\gamma_i)^p}{N\pi}\Gamma^2\left(\!\frac{p\!+\!1}{2}\!\right),  
\end{align}
where $\gamma_0=\frac{|h_0|^2P_s}{\sigma^2_w}$ and $\gamma_1=\frac{|h_1|^2P_s}{\sigma^2_w}$, respectively. It can be seen that the distribution of $T_2$ under hypotheses $\mathcal{H}_0$ and $\mathcal{H}_1$ follows a Chi-square or a Gamma distribution given by
\begin{align}
&  \text{Pr}_{T_2|\mathcal{H}_i}(x) = \frac{1}{\theta^{k_i}_i\Gamma(k_i)}x^{k_i-1}e^{-\frac{x}{\theta_i}},~x \geq 0, \forall i \in \{0,1\}.
\end{align}
with given shape and scale parameters as follows:
\begin{align}
\bar{k}_i&=\frac{\mathbb{E}^2\left\{T_2|\mathcal{H}_i\right\}}{ \text{Var}\left\{T_2|\mathcal{H}_i\right\}}=\frac{N\Gamma^2(\frac{p+1}{2})}{\Gamma(\frac{2p+1}{2})\sqrt{\pi}-\Gamma^2(\frac{p+1}{2})},\nonumber\\
\hat{\theta}_i&\!=\!\frac{\text{Var}\left\{T_2|\mathcal{H}_i\right\}}{\mathbb{E}\left\{T_2|\mathcal{H}_i\right\}}=\frac{2^{p/2}(1\!+\!\gamma_i)^{p/2}}{N} \!\times\!\frac{\Gamma(\frac{2p+1}{2})\sqrt{\pi} \!-\!\Gamma^2(\frac{p+1}{2})}{\Gamma(\frac{p+1}{2})\sqrt{\pi}}.
\end{align}
Given \eqref{P_FF} and the distribution of the gamma function, we can derive the following closed-form expression for the detection threshold:
\begin{equation}\label{search_1}
\tau_2 = F^{-1}_{T_2|\mathcal{H}_0}(1-P_F,\hat{k}_0,\hat{\theta}_0).
\end{equation}
Similarly, by using \eqref{P_DD}, the probability of detection for the IED can be expressed as
\begin{align}\label{P_D_3}
P_D &= 1 - F_{T_2|\mathcal{H}_1}(\tau_2,\hat{k}_1,\hat{\theta}_1)  \nonumber \\ &= 1 - F_{T_2|\mathcal{H}_1}\left(F^{-1}_{T_2|\mathcal{H}_0}(1-P_F,\hat{k}_0,\hat{\theta}_0),\hat{k}_1,\hat{\theta}_1\right).
\end{align}
Specifically, $ F_{T_2|\mathcal{H}_1}(x,\hat{k}_1,\hat{\theta}_1)  = \int_0^{t}   \frac{1}{\hat{\theta}_1^{\hat{k}_1}\Gamma(\hat{k}_1)}x^{\hat{k}_1-1}e^{-\frac{x}{\hat{\theta}_1}}\text{d}x$ represents the cumulative distribution function (CDF) of the Gamma distribution with shape parameter $\hat{k}_1$ and scale parameter $\hat{\theta}_1$. In addition,  $F^{-1}_{T_2|\mathcal{H}_0}(x,\hat{k}_0,\hat{\theta}_0)$ is the inverse CDF with shape parameter $\hat{k}_0$ and scale parameter $\hat{\theta}_0$. A simple line search can numerically obtain the optimal value of $p$ via \eqref{search_1} and \eqref{P_D_3}. This procedure involves numerically searching for the value of $p$ that maximizes the probability of detection for a given false alarm rate. 

\subsubsection{\textbf{TED}}
This is a special case of the IED with $p=2$. It is used to reduce the need for other parameter values and channel knowledge \cite{Kartheek_Devineni_2,Kang_Lu,5208031,6987540}. The test statistic and decision rule for TED can be written as
\begin{equation}\label{8}
T_3 = \frac{1}{N}\sum^{N}_{n=1}  \frac{|y(n)|^2 }{\sigma^2_w}  \vc{\gtreqless}{\mathcal{H}_1}{\mathcal{H}_0}  \tau_3,
\end{equation}
where $\tau_3$ denotes the detection threshold. The PDF of $T_3$ under $\mathcal{H}_0$ and $\mathcal{H}_1$ can be shown to follow a chi-square distribution or a Gamma distribution with shape parameters $k_0=\frac{N}{2}$ and $k_1=\frac{N}{2}$, scale parameters $\theta_0=\frac{2}{N}(1+\gamma_0)$ and $\theta_1=\frac{2}{N}(1+\gamma_1)$, respectively \cite{Fadel}. Using the probability of false alarm in \eqref{P_FF}, the NP rule can be used to determine the detection threshold as follows:
\begin{equation}\label{dete_threshod_17}
\tau_3 = F^{-1}_{T_3|\mathcal{H}_0}(1-P_F,k_0,\theta_0).
\end{equation}
With the help of \eqref{P_DD}, the probability of detection for TED can be derived as below:
\begin{align}\label{P_D_2}
P_D &= 1 - F_{T_3|\mathcal{H}_1}(\tau_3,k_0,\theta_0) \nonumber \\& = 1 - F_{T_3|\mathcal{H}_1}\left(F^{-1}_{T_3|\mathcal{H}_0}(1-P_F,k_0,\theta_0),k_1,\theta_1\right).
\end{align}

\begin{table*}[t]
\centering
\captionsetup{font=small,labelfont={color=DarkGray,bf},textfont={color=DarkGray}}
\caption{Comparison of different detectors in the AmBC system. }
\begin{tabularx}{\textwidth}{
>{\raggedright\arraybackslash}X
>{\raggedright\arraybackslash}X
>{\raggedright\arraybackslash}X
>{\raggedright\arraybackslash}X
>{\raggedright\arraybackslash}X
}
\toprule[1.5pt]
\rowcolor{LightGray}
\textbf{\color{DarkGray} Feature} & \textbf{\color{DarkGray} JCED} & \textbf{\color{DarkGray} IED} & \textbf{\color{DarkGray} TED} & \textbf{\color{DarkGray} ML} \\
\midrule[1pt]
Detection Methodology & Utilizes both energy and signal correlation & Enhances TED by introducing a variable exponent $p$ & Based on the energy of the received signal & Based on the average power of received signal samples \\
\midrule
Computational Complexity & Moderate with optimization required &  Moderate with parameter $p$ search & Low  & Moderate with additional channel estimation \\
\midrule
Optimization Requirement & Requires optimization for best weights & Requires searching for optimal $p$ &  $\times$ &  $\times$ \\
\midrule
CSI Dependency & Low & Low & Low & High \\
\midrule
Implementation Difficulty & Moderate, dependent on optimization & Moderate, dependent on $p$ search method &  Simple & Complex due to channel estimation \\
\bottomrule[1.5pt]
\end{tabularx}
\label{tab:litSummary}
\end{table*}

Finally, Table \ref{tab:litSummary} compares different AmBC signal detection techniques.

\subsubsection{\textbf{IED under generalized noise}}
In this subsection, we evaluate the performance of the IED under diverse noise conditions using the McLeish distribution \cite{mcleish1982robust}, chosen for its compatibility with both Gaussian and non-Gaussian noise. This distribution, characterized by unimodality, symmetry, finite moments, and heavier tails, effectively models environments with noise deviations from ideal Gaussian behavior, typically caused by external disturbances or anomalies.

The generalized noise is represented as $w[\cdot]=\mathcal{CML}(0,\sigma_w^2,q)$, where $\sigma_w^2$ denotes the noise variance, and $q \in \mathbb{R}^{+}$ encapsulates the non-Gaussianity of the noise. Here, $q$ is a crucial parameter, extending the model's scope beyond Gaussian limits. The PDF of the McLeish distribution is formulated as follows \cite[Eq. (85)]{Yilmaz}:
\begin{equation}
f_{w}(w)= \frac{2\sqrt{q}|w|^{q-1}}{\sqrt{2\sigma_w^2}\pi \Gamma(q)} K_{q-1}\left(\sqrt{\frac{2q}{\sigma_w^2}}|w|\right).
\end{equation}
This PDF adaptively transitions across various noise models based on the value of $q$. Specifically, for $q=1$, it aligns with the circularly symmetric complex (CSC) Laplacian distribution, offering a noise model with heavier tails than the Gaussian. As $q$ approaches infinity ($q \rightarrow +\infty$), the distribution converges to the Gaussian model, ideal for standard noise conditions. In the extreme case where $q \rightarrow 0^{+}$, it resembles Dirac’s distribution, representing a highly localized noise model. 

We now  define the test statistic for a sufficiently large number of received signal samples as follows:
\begin{equation}
T_4 = \frac{1}{N}\sum^{N}_{n=1}
|y(n)|^p  \vc{\gtreqless}{\mathcal{H}_1}{\mathcal{H}_0}  \tau_4.
\end{equation}
Although the underlying noise does not appear to be Gaussian, the PDF of $T_4$ closely resembles the shape of the Gaussian distribution. Using \eqref{P_F} and \eqref{P_D}, $\tau_4$ can be calculated using a simple closed-form expression as follows:
\begin{equation}\label{optimal_threhold}
\tau_4 =\mathcal{Q}^{-1}(P_F)\sqrt{\text{Var}(T_4|\mathcal{H}_0)}+\mathbb{E}(T_4|\mathcal{H}_0).
\end{equation}
The  mean and variance of the test statistic under hypothesis $\mathcal{H}_j$, $\forall j \in \{0, 1\}$ are  as follows:
\begin{align}
\mathbb{E}(T_4|\mathcal{H}_j)&=\mathbb{E}\left[\left|h_js(n)+w(n)\right|^{p}\right],\\
\text{Var}(T_4|\mathcal{H}_j)&=\!\frac{1}{N}\Big(\mathbb{E}\left[\left|h_js(n)+w(n)\right|^{2p}\right]\nonumber \\&  \hspace{22mm}-\mathbb{E}\left[\left|h_js(n)+w(n)\right|^{p}\right]^{2}\Big).
\end{align}

\begin{proposition}\label{propro_1}
The exact expression of test statistic moments under different hypotheses for arbitrary $p$ excluding any even integer $p \geq 2$ is represented as
\begin{align}\label{varaince_genralized}
\!\!\!\mathbb{E}\left[\left|h_js(n)+w(n)\right|^{p}\Big|{\mathcal{H}_{j}}\right]&=\frac{\Gamma(1+\frac{p}{2})({|h_j|^2P_s})^{\frac{p}{2}}}{\Gamma(q)\Gamma(-\frac{p}{2})}\nonumber\\&\hspace{-4mm}\times \MeijerG{1,2}{2,1}{ 1-q,~\frac{p}{2}+1 \\ 0 } { \frac{1}{\gamma_j q} }.
\end{align}
As for the second case where $p=\{2,4,6, \ldots$\}, the above equation is found to be\begin{align}\label{varaince_genralized_2}
& \mathbb{E}\left[\left|h_js(n)+w(n)\right|^{p}\Big|{\mathcal{H}_{j}}\right] =\frac{\Gamma(1+\frac{p}{2})}{\Gamma(q)} \nonumber\\&
\hspace{12mm} \times \sum^{p/2}_{k=0} \dbinom{p/2}{k} \left({|h_j|^2P_s}\right)^{\frac{p}{2}-k}\sigma_w^{2k}\Gamma(k+q)q^{-k}.
\end{align}
\end{proposition}
\begin{proof}
Please refer to Appendix \ref{app_2}.
\end{proof}
Furthermore, for the DIC case under hypothesis  ${\mathcal{H}_{0}}$, it can be easily shown that
\begin{equation}
\mathbb{E}\left[\left|w(n)\right|^p \Big| {\mathcal{H}_{0}}\right] = \frac{\Gamma\left(\frac{p}{2} + q\right)\Gamma\left(\frac{p}{2} + 1\right)}{\Gamma(q)q^{p/2}} \sigma_w^p.
\end{equation}
To determine the optimal $p$ value, one must solve the following optimization problem:
\begin{subequations}
\begin{align}\label{P2.2}
\text{(P1)}: & \hspace{1em} \underset{p}{\text{argmax}}  \hspace{1em}  \frac{\tau_4-\mathbb{E}(T_4|\mathcal{H}_1)}{\sqrt{\text{Var}(T_4|\mathcal{H}_1)}},   \quad \text{s.t.} \hspace{1em} p>0,
\end{align} 
\end{subequations}
which can be numerically computed by a simple line search.

\subsubsection{\textbf{JCED under generalized noise}}
Next, we will study JCED under generalized noise. According to Appendix \ref{app_2}, let's consider $w(n) = \sqrt{G(n)}C(n)$, where $C(n) \sim \mathcal{CN}(0,\sigma_w^2)$ and $G(n)$ is gamma distributed. However, the variance of $Z_1$ can be presented as follows:
\begin{align}
\text{Var}(Z_1) &= \left\{ \begin{array}{l}
\left(N\left(1 + \frac{2}{q}\right)+2\gamma_0\right)\sigma_w^4,\quad \text{if} \:\:\mathcal{H}_0 ,\\
\left(N\left(1 + \frac{2}{q}\right)+2\gamma_1\right)\sigma_w^4,\quad \text{if} \:\:\mathcal{H}_1,
\end{array} \right.
\end{align}
where the variance of $\Re\{w^*(n)\}$ and $|w(n)|^2$ is calculated as
\begin{align}
\text{Var}(\Re\{w^*(n)\}) &= \mathbb{E}\left[\left\{\sqrt{G(n)} \Re\{C^*(n)\}\right\}^2\right] \nonumber\\
&= \mathbb{E}\left[G(n)\right] \cdot \mathbb{E}\left[\left\{\Re\{C^*(n)\}\right\}^2\right]= \frac{\sigma_w^2}{2},\\
\text{Var}(|w(n)|^2) &= \mathbb{E}\left[G(n)^2|C(n)|^4\right] - \left(\mathbb{E}\left[G(n)|C(n)|^2\right]\right)^2 \nonumber \\
&= \left(1 + \text{Var}(G(n))\right)2\sigma_w^4 - \left(\sigma_w^2\right)^2\nonumber \\
&= \left(1 + \frac{1}{q}\right)2\sigma_w^4 - \sigma_w^4= \left(1 + \frac{2}{q}\right)\sigma_w^4,
\end{align}
respectively. Also, the mean and variance of $Z_2$ and covariance between $Z_1$ and $Z_2$ are equivalent to \eqref{mean_R} and \eqref{covarince_correlatior}, respectively.

\section{Detector Design: Multi-Antenna Reader}
This study primarily focuses on a single-antenna model to conduct a detailed analysis of the JCED and IED, laying a solid foundation for understanding their core functionalities. However, we also incorporate a multi-antenna reader to enhance signal processing capabilities \cite{Ma2015}. This approach paves the way for future research expansions into multi-antenna RF sources and multi-tag environments, where we aim to explore advanced interference cancellation and signal detection techniques.

As a simple scenario, $M$ antennas are assumed to be located at spatially independent positions, so $M$ degree-of-freedom can be fully exploited during the ED process \cite{Wang6388431,Kalamkar2015ImpactOA}.

\subsubsection{\textbf{JCED}}
The test statistic for this detector is constructed using a linear combination of the energy of the samples and first-order correlation values of the received signals given by
\begin{align} 
T_5&=\alpha \sum^{N-1}_{n=0}\sum^{M}_{m=1}|y_m(n)|^2+\beta \sum^{N-2}_{n=0}\sum^{M}_{m=1} y_m(n+1) y_m^*(n) \nonumber \\&= \alpha Z_1 + \beta Z_2.
\end{align}
According to the CLT \cite{gnedenko1954limit}, we can conclude that for large values of $N$, the statistic $Z_1$ is expected to have a normal distribution. Therefore, the resulting mean and variance of $Z_1$ are given by
\begin{align}\label{522}
\mathbb{E}(Z_1) &= \left\{ \begin{array}{l}
M (E_s\sigma_0^2+N\sigma_w^2),\quad \text{if} \:\:\mathcal{H}_0 ,\\
M (E_s\sigma_1^2+N\sigma_w^2),\quad \text{if} \:\:\mathcal{H}_1,
\end{array} \right.\\[8pt] \nonumber
\text{Var}(Z_1) &= \left\{ \begin{array}{l}
M(2E_s\sigma_0^2+N\sigma_w^4),\quad \text{if} \:\:\mathcal{H}_0 ,\\
M(2E_s\sigma_1^2+N\sigma_w^4),\quad \text{if} \:\:\mathcal{H}_1.
\end{array} \right.
\end{align}
Based on \eqref{received_signal_2}, $\sigma^2_0$ and $\sigma^2_1$ are variances of $h^0_m$ and $h^1_m$, respectively. Consequently, $Z_2$ has approximately a normal distribution with a mean and variance of
\begin{align}
\mathbb{E}(Z_2) &= \left\{ \begin{array}{l}
M\sigma_0^2R_{ss}(1),\quad \text{if} \:\:\mathcal{H}_0 ,\\
M\sigma_1^2R_{ss}(1),\quad \text{if} \:\:\mathcal{H}_1,
\end{array} \right.  \\ \nonumber
\text{Var}(Z_2)& = \left\{ \begin{array}{l}
M(N-1)\sigma_w^4+2ME_s\sigma_0^2\sigma_w^2,\quad \text{if} \:\:\mathcal{H}_0 ,\\
M(N-1)\sigma_w^4+2ME_s\sigma_1^2\sigma_w^2,\quad \text{if} \:\:\mathcal{H}_1,
\end{array} \right.
\end{align}  
respectively. The covariance between $Z_1$ and $Z_2$ can be represented as
\begin{equation} 
\text{Cov}(Z_1, Z_2) = \left\{ \begin{array}{l}
2M\sigma_w^2\sigma_0^2 R_{ss}(1),\quad \text{if} \:\:\mathcal{H}_0,\\
2M\sigma_w^2\sigma_1^2 R_{ss}(1),\quad \text{if} \:\:\mathcal{H}_1. 
\end{array} \right.
\end{equation}
The mean and variance of the test statistic $T_5$ under different hypotheses can then be stated as
\begin{align}
\mathbb{E}(T_5) &= \left\{ \begin{array}{l}
\alpha M (E_s\sigma_0^2+N\sigma_w^2)  +\beta M\sigma_0^2R_{ss}(1) , \quad \text{if} \:\:\mathcal{H}_0 ,\\
\alpha M  (E_s\sigma_1^2+N\sigma_w^2)  +\beta M\sigma_1^2R_{ss}(1), \quad \text{if} \:\:\mathcal{H}_1,
\end{array} \right.\\
\text{Var}(T_5)&= \left\{ \begin{array}{l}
\mathbf{w}^T \Sigma_{\mathcal{H}_0} \mathbf{w}, \quad \text{if} \:\:\mathcal{H}_0 ,\\
\mathbf{w}^T \Sigma_{\mathcal{H}_1} \mathbf{w}, \quad \text{if} \:\:\mathcal{H}_1, 
\end{array} \right.
\end{align}
where $\mathbf{w}= (\alpha,~\beta)^T$ includes the associate weights and
\begin{align}
&	\Sigma_{\mathcal{H}_j}\!\! = \!\!\begin{bmatrix}
M(2E_s\sigma_j^2\sigma_w^2\!+\!N\sigma_w^4)&  	2M\sigma_w^2\sigma_j^2 R_{ss}(1)\\
2M\sigma_w^2\sigma_j^2 R_{ss}(1) & M(N\!-\!1)\sigma_w^4\!+\!2ME_s\sigma_j^2\sigma_w^2,
\end{bmatrix}\!\!,      
\end{align}
$\forall j \in \{0, 1\}$. Similar to \eqref{P2.1}, we can maximize the probability of detection for a given probability of false alarm. 

\subsubsection{\textbf{IED}}
The test statistic and decision rule for the IED can be stated as follows:
\begin{equation} 
T_6 =  \sum^{N}_{n=1} \sum^{M}_{m=1}\left(\frac{|y_m(n)|}{\sigma_w}\right)^p  \vc{\gtreqless}{\mathcal{H}_1}{\mathcal{H}_0}  \tau_6.
\end{equation}
where $\tau_6$ is the detection threshold to be determined. Based on the CLT \cite{gnedenko1954limit}, the mean and variance of $T_6, \forall i \in \{0,1\},$ can be expressed as
\begin{align}
\mathbb{E}\left\{T_6|\mathcal{H}_i\right\} &= \frac{M2^{p/2}}{\sqrt{\pi}}\Gamma\left(\frac{p+1}{2}\right)\left(1+\gamma_i\right)^{p/2},\nonumber\\
\text{Var}\left\{T_6|\mathcal{H}_i\right\} &= \frac{M2^{p}(1+\gamma_i)^p\Gamma(\frac{2p+1}{2})}{\sqrt{\pi}}\nonumber\\&\hspace{10mm}-\frac{M^22^p(1+\gamma_i)^p}{\pi}\Gamma^2\left(\frac{p+1}{2}\right),  
\end{align}
Using \eqref{P_F} and \eqref{P_D}, the probability of detection can be obtained as
\begin{equation}\label{p_d_iedm}
P_D\! = \!Q \!\left(\!\!\frac{\sqrt{\text{Var}(T_6|\mathcal{H}_0)}Q^{-1}(P_F)\!+\!	\mathbb{E}\left\{T_6|\mathcal{H}_0\right\}\!-\!	\mathbb{E}\left\{T_6|\mathcal{H}_1\right\} }{\sqrt{\text{Var}(T_6|\mathcal{H}_1)}}\!\!\right)\!\!.
\end{equation}

\subsubsection{\textbf{TED}}
The decision rule of the ED can be written as
\begin{equation}
T_7 =  \sum^{N}_{n=1}\sum^{M}_{m=1} \left|{y_m(n)}\right|^2  \vc{\gtreqless}{\mathcal{H}_1}{\mathcal{H}_0}  \tau_7.
\end{equation}
Generally, the statistic $T_7$ is Chi-squared distributed with $2MN$ degrees of freedom since it is the sum of the squared Gaussian random variables. According to the CLT \cite{gnedenko1954limit}, we can conclude that for large values of $N$,  the statistic $T_7$ is expected to have a normal distribution. The mean and variance of $T_7$ are given by
\begin{align}
\mathbb{E}(T_7) &= \left\{ \begin{array}{l}
MN (P_s\sigma_0^2+\sigma_w^2),\quad \text{if} \:\:\mathcal{H}_0 ,\\
MN (P_s\sigma_1^2+\sigma_w^2),\quad \text{if} \:\:\mathcal{H}_1,
\end{array} \right.\\[8pt] \nonumber
\text{Var}(T_7) &= \left\{ \begin{array}{l}
MN(P_s\sigma_0^2+\sigma_w^2)^2,\quad \text{if} \:\:\mathcal{H}_0 ,\\
MN(P_s\sigma_1^2+\sigma_w^2)^2,\quad \text{if} \:\:\mathcal{H}_1.
\end{array} \right.
\end{align}
Consequently, the probability of detection is given by \eqref{p_d_iedm}.

\section{Deriving the AUC of ROC}
Although the ROC curve provides a comprehensive evaluation of detector performance, there are cases where a single metric is adequate for assessing effectiveness. In such instances, the area under the curve (AUC) is a suitable evaluation metric. The AUC is widely recognized as a valid measure of detection capability, supported by the area Theorem \cite{wickens2001elementary} and research findings in \cite{Syed_Safwan} and \cite{hanley1982meaning}. We derive a general analytical formula for the AUC to determine the performance of IED and TED. According to the CLT \cite{gnedenko1954limit},  $P_D$ and $P_F$ expressions  for different detectors can be stated  as
\begin{align}
&P_F = \Pr(W>\lambda|\mathcal{H}_0)  = \mathcal{Q}\left(\frac{\lambda-\mathbb{E}\left\{W|\mathcal{H}_0\right\}}{\sqrt{{\text{Var}\left\{W|\mathcal{H}_0\right\}}}} \right),\\ \nonumber
&P_D = \Pr(W>\lambda|\mathcal{H}_1)  = \mathcal{Q}\left(\frac{\lambda-\mathbb{E}\left\{W|\mathcal{H}_1\right\}}{\sqrt{{\text{Var}\left\{W|\mathcal{H}_1\right\}}}} \right),
\end{align}
where $\lambda\in\{\tau_i,~i=2,\ldots,4\}$ and $W\in\{T_i,~i=2,\ldots,4\}$. The area covered by the ROC is 
\begin{equation}\label{588}
A(\gamma_j) = \int^1_0 P_D(\lambda,\gamma_j)\text{d}P_F(\lambda), j \in \{0, 1\}. 
\end{equation}
Eq. \eqref{588}  can also be expressed as an integral  with respect to $\lambda$, as shown below:
\begin{equation}\label{AUC}
A(\gamma_j)=-\int^{{\infty}}_{-\infty} P_D(\lambda,\gamma_j)\frac{\partial P_F(\lambda)}{\lambda} \text{d}\lambda,
\end{equation}
where the partial derivative $\frac{\partial P_F(\lambda)}{\lambda}$ in \eqref{AUC} can be evaluated as
\begin{equation}\label{566}
\frac{\partial P_F(\lambda)}{\lambda}  =  \frac{\partial }{\lambda} \frac{1}{2\pi}\int^{\infty}_{\frac{\lambda-\mathbb{E}\left\{W|\mathcal{H}_0\right\} }{\sqrt{ \text{Var}\left\{W|\mathcal{H}_0\right\}}}}\text{exp}(-\frac{t^2}{2})\text{d}t.
\end{equation}
By introducing a change of variable $x = \sqrt{{\text{Var}\left\{W|\mathcal{H}_0\right\}}}t+\mathbb{E}\left\{W|\mathcal{H}_0\right\} $, \eqref{566} can be rewritten as
\begin{align}\label{derivative}
\frac{\partial P_F(\lambda)}{\lambda} \!\! &= \!\!\sqrt{\frac{1/2\pi}{ \text{Var}\left\{W\!|\!\mathcal{H}_0\right\}} }  \frac{\partial }{\lambda} \int^{\infty}_{\lambda}\!\!\text{exp}\left(\!\!-\frac{\left(x\!-\!\mathbb{E}\left\{W\!|\!\mathcal{H}_0\right\}\right)^2}{2\text{Var}\left\{W\!|\!\mathcal{H}_0\right\}}\!\right)\text{d}x\nonumber\\
&=-\sqrt{\frac{1/2\pi }{\text{Var}\left\{W|\mathcal{H}_0\right\}} }\text{exp}\left(-\frac{\left(\lambda-\mathbb{E}\left\{W|\mathcal{H}_0\right\}\right)^2}{2\text{Var}\left\{W|\mathcal{H}_0\right\}}\right).
\end{align}
\begin{figure*}
\centering
\begin{minipage}[b]{.47\textwidth}
\centering
\includegraphics[width=3.15in]{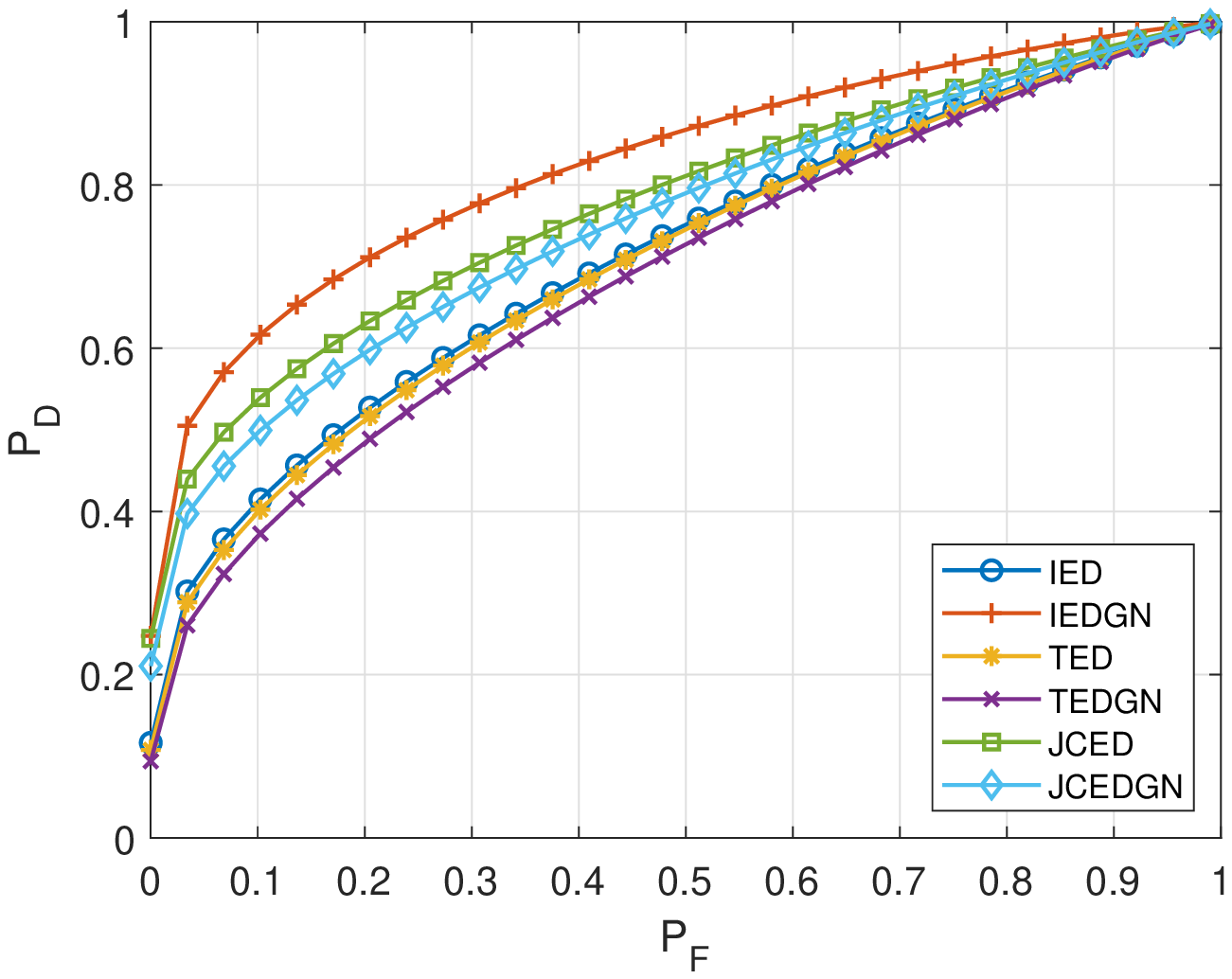}
\subcaption{With DIC technique} \label{ROC_0_1}
\end{minipage}\qquad
\begin{minipage}[b]{.47\textwidth}
\centering
\includegraphics[width=3.15in]{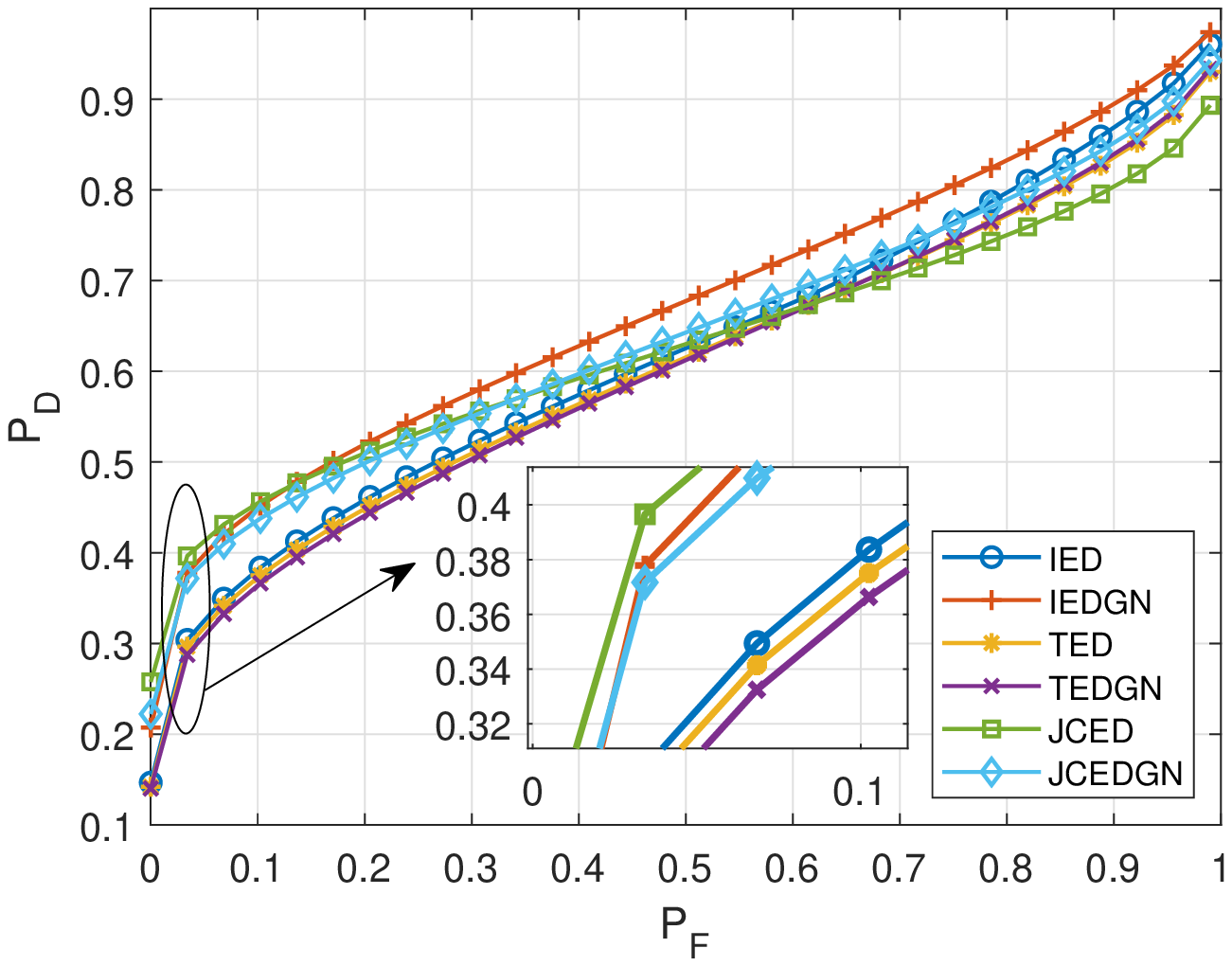}
\subcaption{Without DIC technique} \label{ROC_0_2}
\end{minipage}
\caption{A comparison of the ROC curves for different detectors with/without DIC technique when $P_s=10$ dBm and $M=1$.}
\label{ROC_0}
\end{figure*}
\begin{figure*}[t]
\centering
\begin{minipage}[b]{.47\textwidth}
\centering
\includegraphics[width=3.15in]{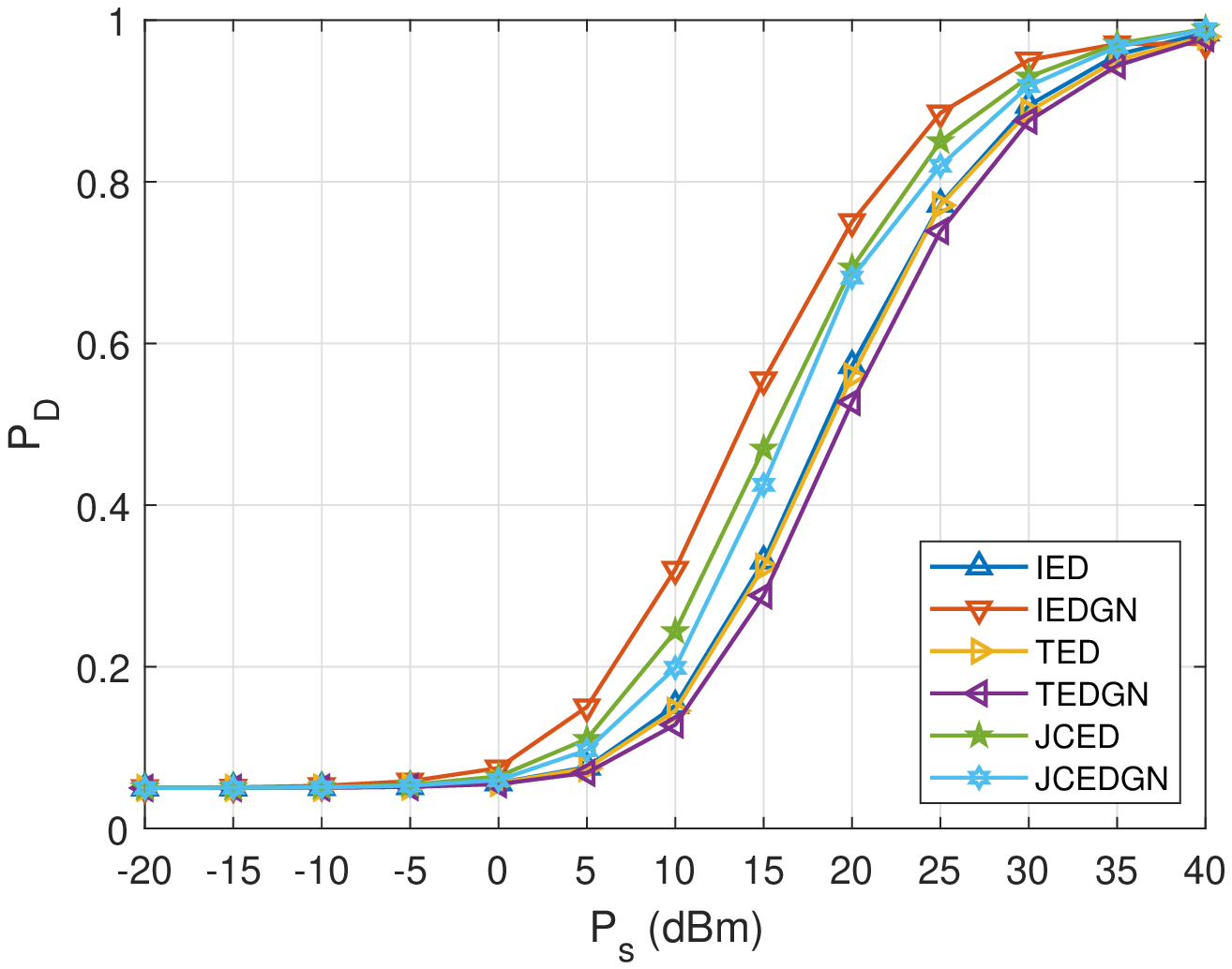}
\subcaption{With DIC technique.} \label{pd_snr_0}
\end{minipage}\qquad
\begin{minipage}[b]{.47\textwidth}
\centering
\includegraphics[width=3.15in]{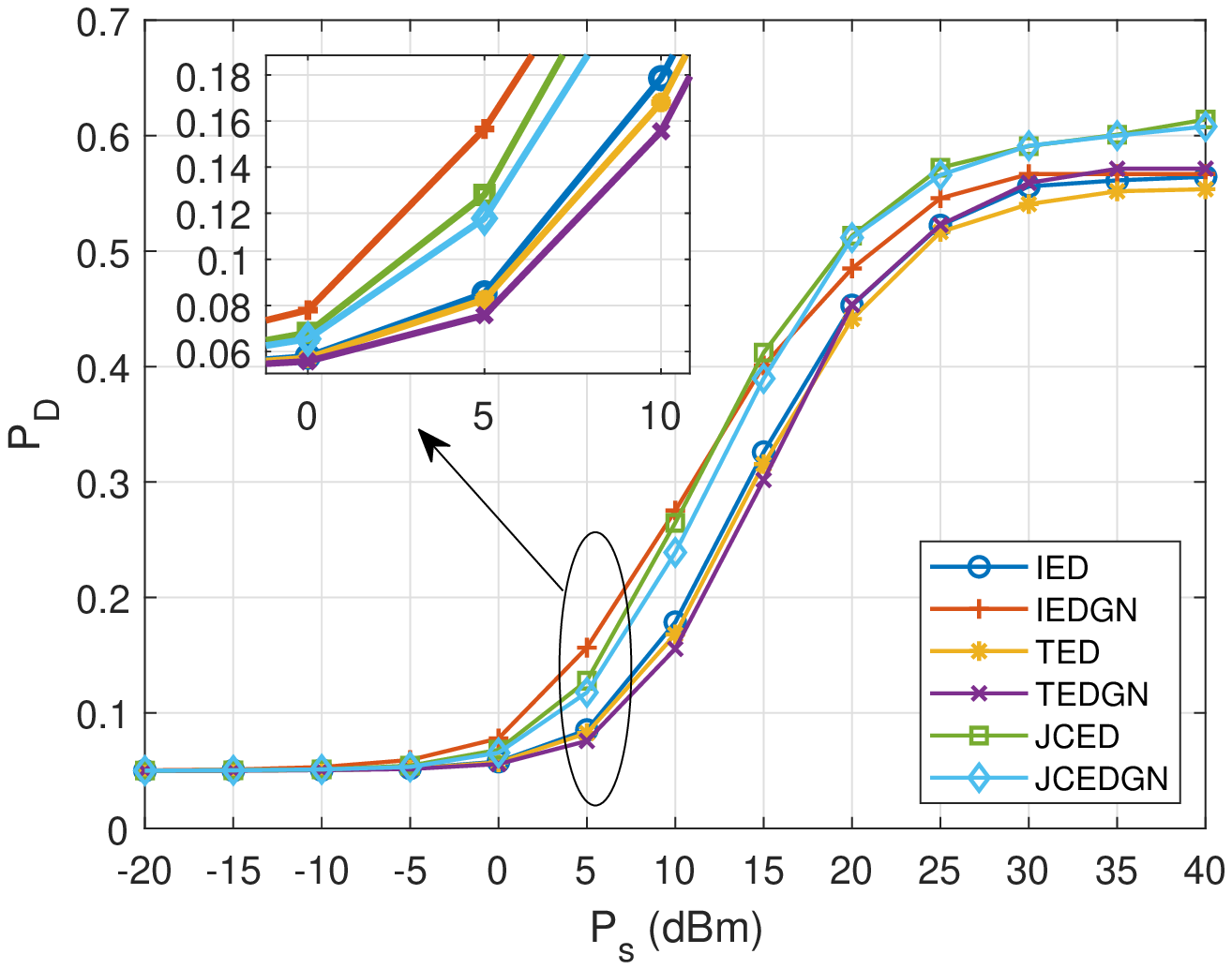}
\subcaption{Without DIC technique} \label{pd_snr_1}
\end{minipage}
\caption{Probability of detection vs. various $P_s$ regions for different detectors with/without DIC technique  when $M=1$.\vspace{-5mm}}
\label{pd_snr}
\end{figure*}
\begin{figure*}[t]
\begin{minipage}[h]{0.47\linewidth}
\begin{center}
\includegraphics[width=1\linewidth]{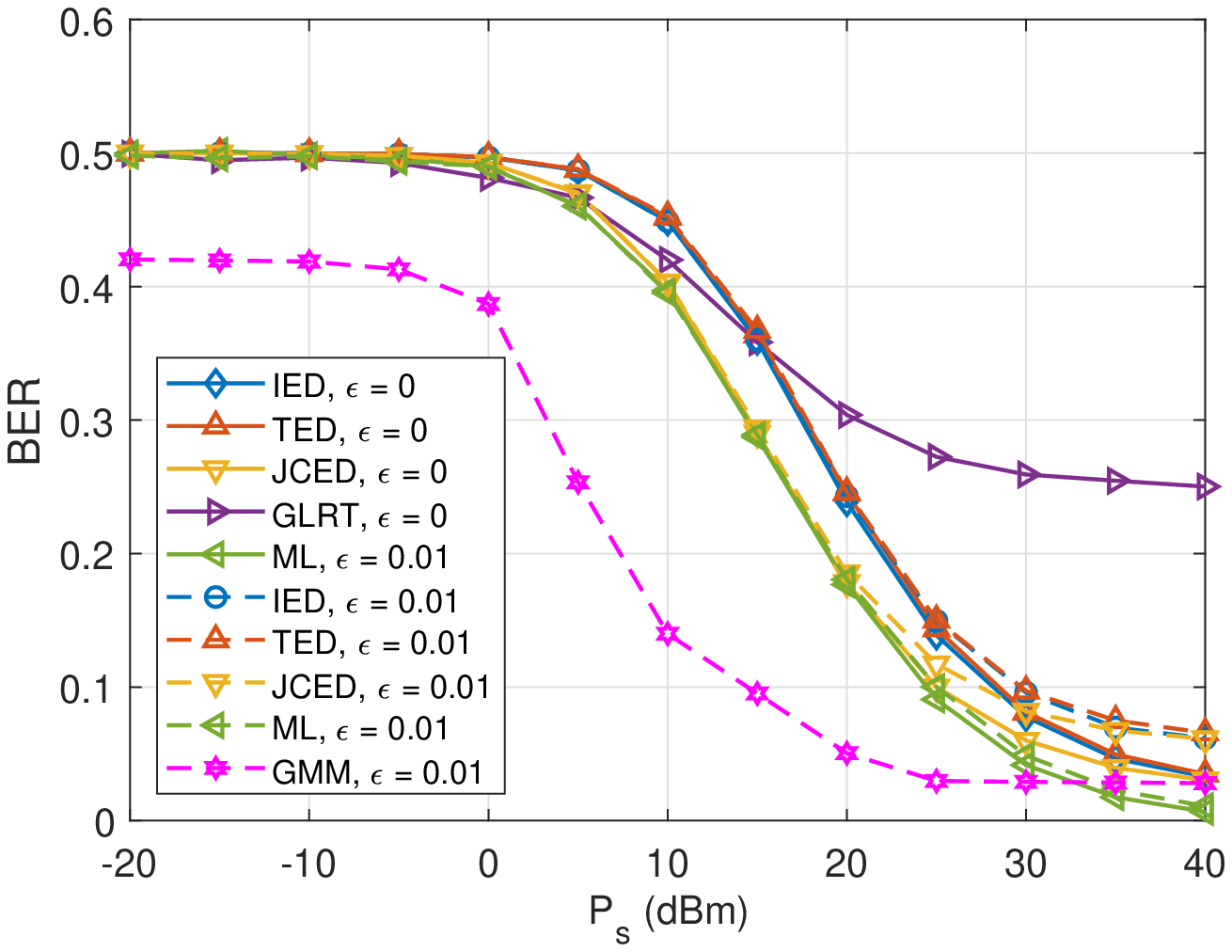} 
\subcaption{{With  DIC technique.} }
\label{BER_sic}
\end{center} 
\end{minipage}
\hfill
\vspace{0.2 cm}
\begin{minipage}[h]{0.47\linewidth}
\begin{center}
\includegraphics[width=1\linewidth]{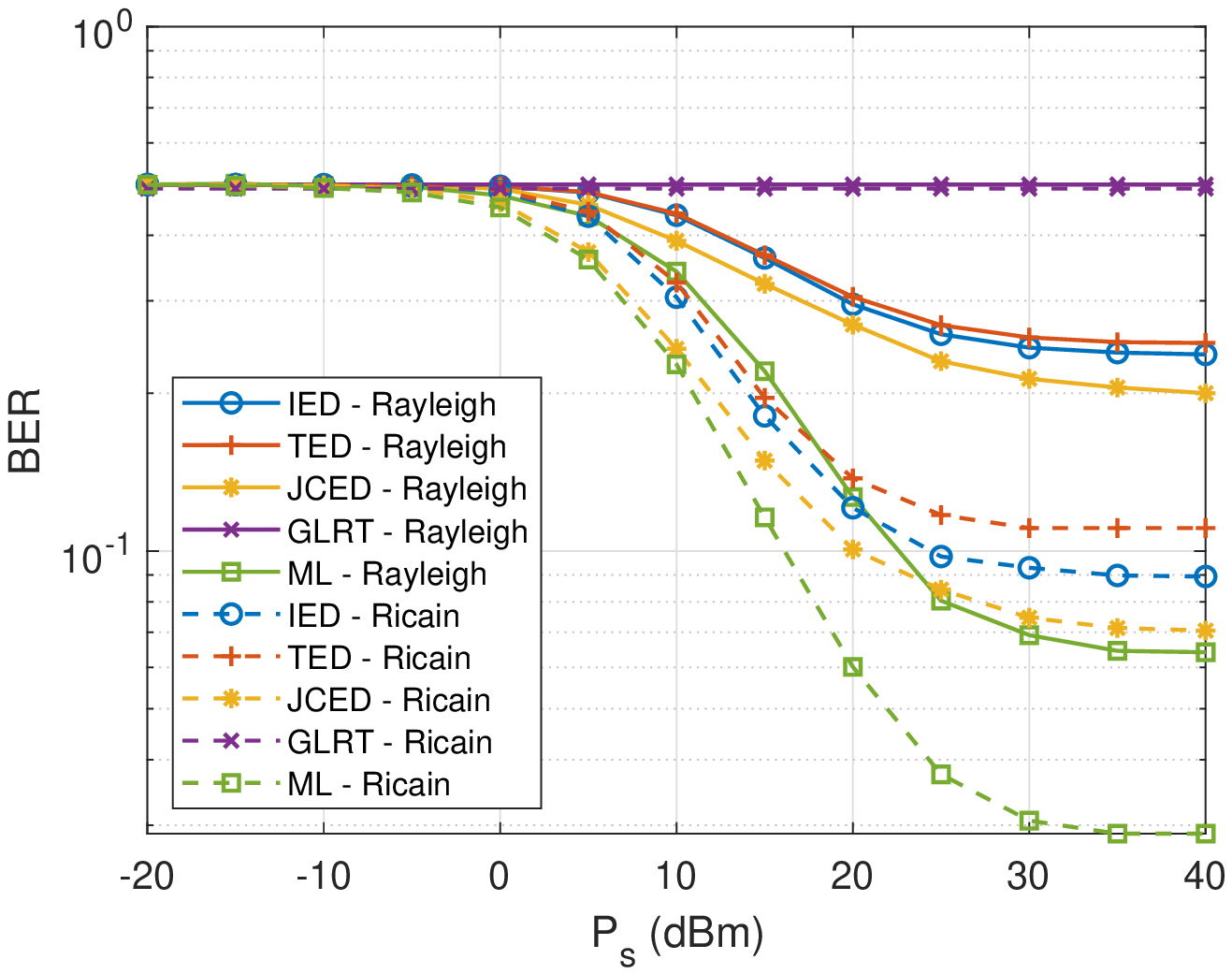} 
\subcaption{Without DIC technique.} \label{BER}
\end{center}
\end{minipage}
\caption{BER vs. various $P_s$ regions for different detectors with/without DIC technique, including RI  when $M=1$. \vspace{-5mm}} \label{BERR}
\label{ris}
\end{figure*}
Now putting \eqref{derivative} in \eqref{AUC} yield:
\begin{align}\label{AUC_gg}
A(\gamma_j) &=\sqrt{\frac{1/2\pi}{\text{Var}\left\{W|\mathcal{H}_0\right\}}}\int^{\infty}_{-\infty}\mathcal{Q}\left(\frac{\lambda-\mathbb{E}\left\{W|\mathcal{H}_1\right\}}{\sqrt{{\text{Var}\left\{W|\mathcal{H}_1\right\}}}}\right)\nonumber\\& \times \text{exp}\left(-\frac{\left(\lambda-\mathbb{E}\left\{W|\mathcal{H}_0\right\}\right)^2}{2\text{Var}\left\{W|\mathcal{H}_0\right\}}\right)\text{d}\lambda.
\end{align}
\begin{proposition}\label{propro_3}
The integral in \eqref{AUC_gg} is evaluated as $A(\gamma_j) = \mathcal{Q}\left(\frac{a}{\sqrt{b^2+1}}\right),$
where
\begin{align}
&a=\frac{\mathbb{E}\left\{W|\mathcal{H}_0\right\}-\mathbb{E}\left\{W|\mathcal{H}_1\right\}}{\sqrt{{\text{Var}\left\{W|\mathcal{H}_1\right\}}}}   \quad \text{and} \quad b=\sqrt{\frac{\text{Var}\left\{W|\mathcal{H}_0\right\}}{\text{Var}\left\{W|\mathcal{H}_1\right\}}}.   
\end{align}
\begin{proof}
Please refer to Appendix \ref{Pro_3}.
\end{proof}
\end{proposition}

\section{Simulation Results}\label{sim_res}
We next present the performance of the proposed detectors and several benchmarks using Monte Carlo simulations. The following benchmark detectors are considered:  1) IED; 2) IED under generalized noise (IEDGN); 3) TED {\cite{Sudarshan_Guruacharya, Jing_Qian_4,Zeng_Tengchan,Kartheek_Devineni,Devineni_Kartheek_2}}; 4) TED under generalized noise (TEDGN); 5) JCED; 6) GLRT {\cite{Sudarshan_Guruacharya}}; 7) ML {\cite{Jing_Qian,Zargari10320395}}.

The generalized noise model assumes $q=1$, representing Laplace noise; all channel links experience Rayleigh flat fading. Large-scale fading describes the reduction of signal strength as it travels from one location to another, offering insights into the communication range. Factors influencing it include distance, operating frequency, and atmospheric and environmental conditions such as indoor or outdoor settings, urban or rural areas, and open or wooded environments. One fundamental model is the free space pathloss model, which is given by
\begin{equation}
L_{\text{path}} = 32.45 + 20 \log_{10}(d) + 20 \log_{10}(f),
\end{equation}
where $d$ (in km) denotes the distance and $f$ (in MHz) represents the operating frequency. 

The carrier frequency and bandwidth are $915$ MHz and $10$ MHz, respectively. The noise power is assumed to be $-174 \, \text{dBm/Hz}$ \cite{Diluka}. We assume that the transmitted primary signal is $s(n) = 1$ for simplicity. The RF  source antenna has a gain of $6 \, \text{dB}$, the reader antenna a gain of $3 \, \text{dB}$, and the tag antenna a gain of $2 \, \text{dB}$. In addition, we assume that the RF source is positioned $6 \, \text{meters}$ away from the tag and $4 \, \text{meters}$ away from the reader, while the tag is $0.5 \, \text{meters}$ away from the reader. These values are taken from a practical study  \cite{Kellogg}.  The false alarm probability is set at $P_F=0.05$. We also assume a tag reflection coefficient of $\xi=1$ and $N=512$ samples. Each simulation point is derived by averaging over $10^4$ Monte Carlo instances.

The ROC curves of the proposed detectors are compared to that of the TED for different $P_s$ values with and without  DIC in Fig. \ref{ROC_0}. DIC increases $P_D$ compared to no DIC, eliminating direct interference and improving detection accuracy. The optimized-$p$ based IED outperforms TED in all cases. The difference becomes more apparent when $P_F\leq 0.1$, relevant for practical applications with a lower false alarm rate. TED's performance is reduced in Laplacian noise, but TED performs similarly under AWGN and Laplacian at high transmit power regions. The absolute value detector (when $0<p\leq 1$) is best suited for Laplacian noise.  Thus,  IED performs best with Laplacian noise compared to other schemes, highlighting the importance of optimizing $p$ for different scenarios.  In Fig. \ref{ROC_0}, with a false alarm rate of $1\%$, JCED and IED have $22.97\%$ and $5.41\%$ higher $P_D$ than TED, respectively. With DIC, the gains are $34.92\%$ and $3.7\%$, respectively. 

The figure also shows the impact of the correlation between transmitted signal samples on the detection performance. Incorporating sample correlation enhances $P_D$ and allows JCED to outperform IED and TED, highlighting the superior performance achieved when adjacent received signal samples are correlated. Furthermore, it appears that the JCEDGN has a competitive performance compared to the other detectors. Furthermore, the ROC curve of the JCEDGN stands out prominently, indicating a superior $P_D$ relative to both TED and IED. When the DIC technique is applied, the JCEDGN  seems to perform similarly to JCED, indicating robustness in performance with or without DIC.

The performance of different detectors is presented in Fig. \ref{pd_snr} for different $P_s$ values. This depiction offers a comparative analysis of the detector efficiencies. A deeper dive into the data reveals nuanced performance disparities. For instance, the IED has a slight edge over the TED in scenarios with higher $P_s$ values. Both JCED and JCEDGN display notable resistance to Laplacian noise, enhancing their effectiveness in noisy environments and ensuring reliable data transmission. The JCEDGN maintains high detection sensitivity and accuracy compared to TED and IED, which is vital for systems operating in complex noise landscapes. However, it still has worse performance compared to the JCED. The newly introduced JCED is the top performer in this evaluation. The JCED capitalizes on an advanced algorithm that deftly leverages the correlation within the received signal samples. This feature becomes especially pivotal when managing dense signal environments, ensuring reduced error probabilities and a smoother data interpretation. One of the standout observations pertains to the debilitating effect of Laplacian noise on detectors, which is especially evident in the TED. Such noise introduces higher variance and skewness into the signal processing phase, diminishing the TED's accuracy and reliability. The spectral characteristics of this noise, characterized by its heavier tails, challenge TED's foundational algorithms. In stark contrast, both JCED and IEDGN demonstrate admirable adaptability in the face of Laplacian noise. Such inherent resistance makes them prime candidates for deployment in environments where Laplacian noise interference is anticipated, ensuring both data fidelity and transmission reliability.

Furthermore, in Fig. \ref{BERR}, the performance evaluation under varying channel models and RI strength is depicted.  The Rician fading channel model is given by
\begin{equation}
h = \sqrt{\frac{\kappa}{\kappa+1}} h^{\text{LoS}}_i + \sqrt{\frac{1}{\kappa+1}} h^{\text{NLoS}}_i,~i\in\{sr, tr, st\},
\end{equation}
where $\kappa=3$ is the Rician factor and $h^{\text{LoS}}_i=1$ is the deterministic LoS component that corresponds to the direct path between the transmitter and receiver, without any obstructions or scattering \cite{Zargari10320395}. Also, $h^{\text{NLoS}}_i$ is the non-LoS component that follows the Rayleigh fading model.

The integration of this model becomes indispensable in environments characterized by a pronounced LoS path, offering a stark contrast to the Rayleigh model, which predominantly encapsulates scenarios of multi-path propagation absent of a direct path. Due to the existence of the LoS component, the BER is significantly reduced. This reduction in BER is mainly due to reduced vulnerability to interference and fading in the Rician model's direct LoS path. This provides a resilient communication link, capable of mitigating diverse transmission impediments.

Moreover, as observed in Fig. \ref{BER_sic}, an increase in the RI strength correlates with a decline in BER performance. The results in both figures demonstrate that for $P_s$ values below $-10$ dBm, there is negligible performance difference between the proposed detectors and TED. However, as the $P_s$ increases, the performance of IED surpasses TED in both DIC and without DIC scenarios, highlighting the significance of optimizing parameter $p$  for IED. Moreover, JCED consistently outperforms both IED and TED regarding BER across all cases, indicating that incorporating the first-order correlation between received samples reduces BER.

Further, we assessed the performance of the  GLRT detector \cite{Sudarshan_Guruacharya}. The GLRT estimates channel parameters using maximum log-likelihood and then performs the hypothesis test but relies on more accurate knowledge of transmit signal statistics and exhibits degraded detection performance due to sensitivity to uncertainties. Lastly, we analyzed the performance of the ML detector \cite{Jing_Qian_4}, which outperforms all other schemes. {We also compared our proposed method against the GMM, which is an unsupervised learning algorithm {\cite{Qianqian_Zhang_1}}. It offers a superior BER performance compared to the other evaluated techniques.}  

\begin{figure}[t]
\centering
\includegraphics[width=3.15in]{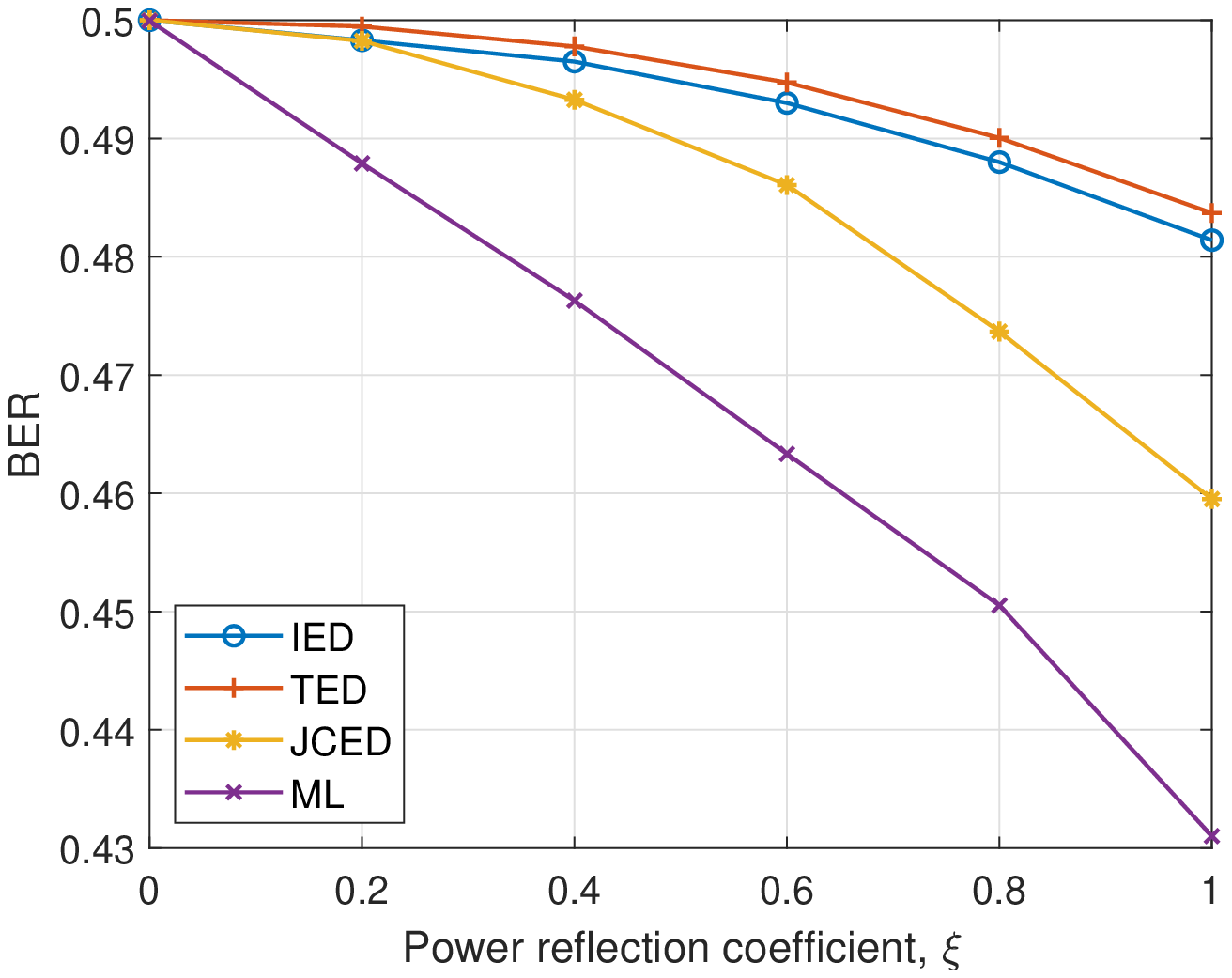}
\caption{BER versus reflection coefficient of the tag with fixed $P_s=5$ dBm  and $M=1$. \vspace{-5mm}} \label{BER_reflection_coff} 
\end{figure}

\begin{figure}
\centering
\includegraphics[width=3.15in]{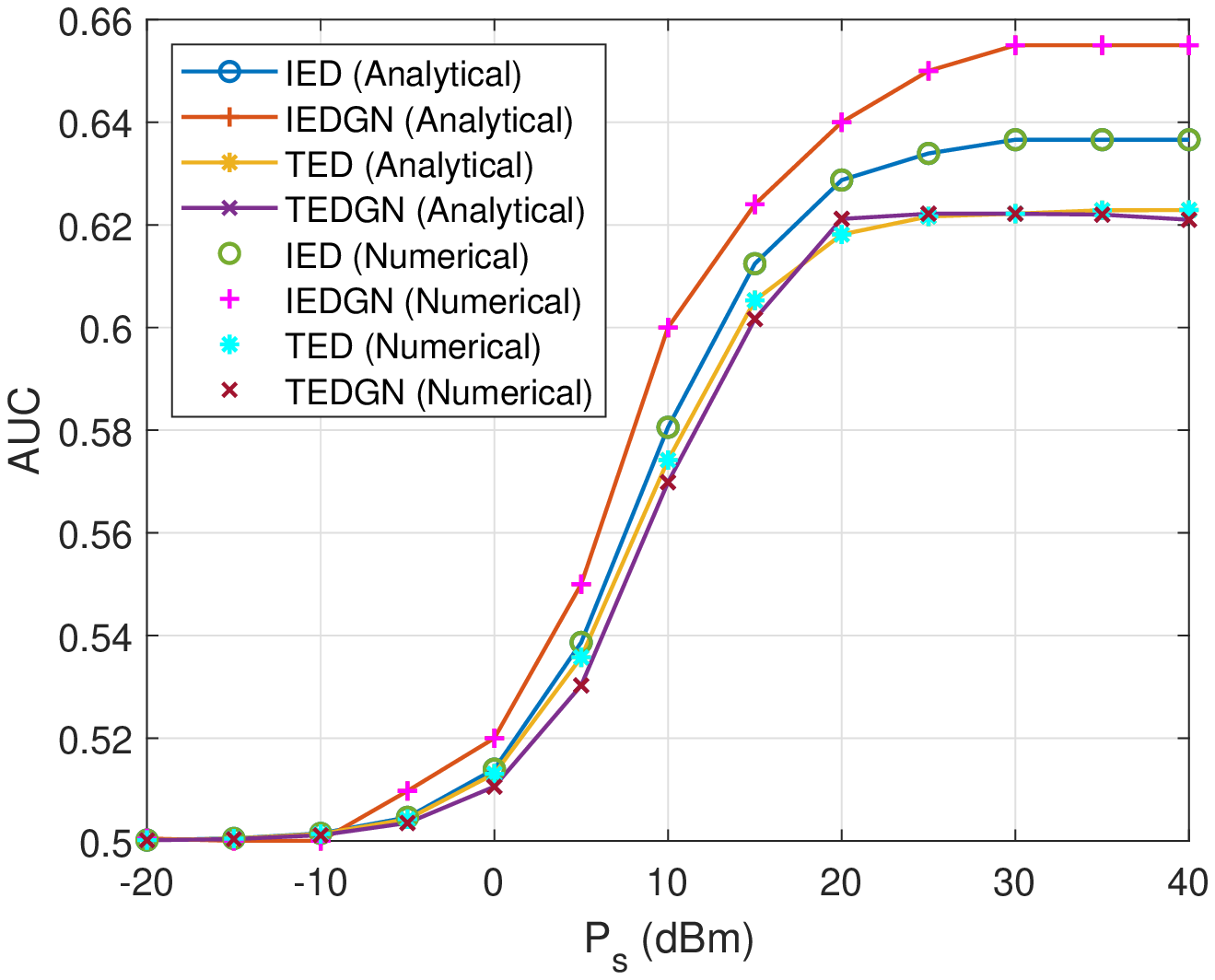}
\caption{AUC vs. various $P_s$ for different detectors when $M=1$. \vspace{-5mm} } \label{AUC_fig}
\end{figure}

Fig. \ref{BER_reflection_coff} demonstrates the influence of the power reflection coefficient, $\xi$, on the BER performance of various detection schemes. It is observed that as $\xi$ increases, there is a notable decrease in BER across all detection schemes. This improvement in BER is attributed to enhanced signal quality, a direct result of a higher proportion of the incident RF wave being reflected for data transmission. Enhanced signal reflection improves the detectors' ability to distinguish between transmitted symbols accurately. On the other hand, although a lower value of $\xi$ increases the energy available for EH, it reduces signal quality. This is evident in the figure where, at lower $\xi$ values, the performance of the detectors converges, reflecting the challenge in symbol distinction due to poorer received signal quality.

This figure highlights the influence of $\xi$ on BER, illustrating a critical trade-off between signal quality for reliable data transmission and sufficient EH to meet backscatter tag circuit sensitivity, often around $-20$ dBm \cite{Diluka,rfidworld2019}. A higher $\xi$ can reduce the energy available for EH, potentially leading to operational issues in tags due to power shortages. This intersection suggests an avenue for future research, particularly in developing advanced detection algorithms that optimize EH efficiency, possibly through machine learning to adapt to fluctuating RF environments, ensuring stable tag function and consistent data transmission.

\begin{figure}
\centering
\begin{subfigure}[b]{0.5\textwidth}
\centering
\includegraphics[width=3.15in]{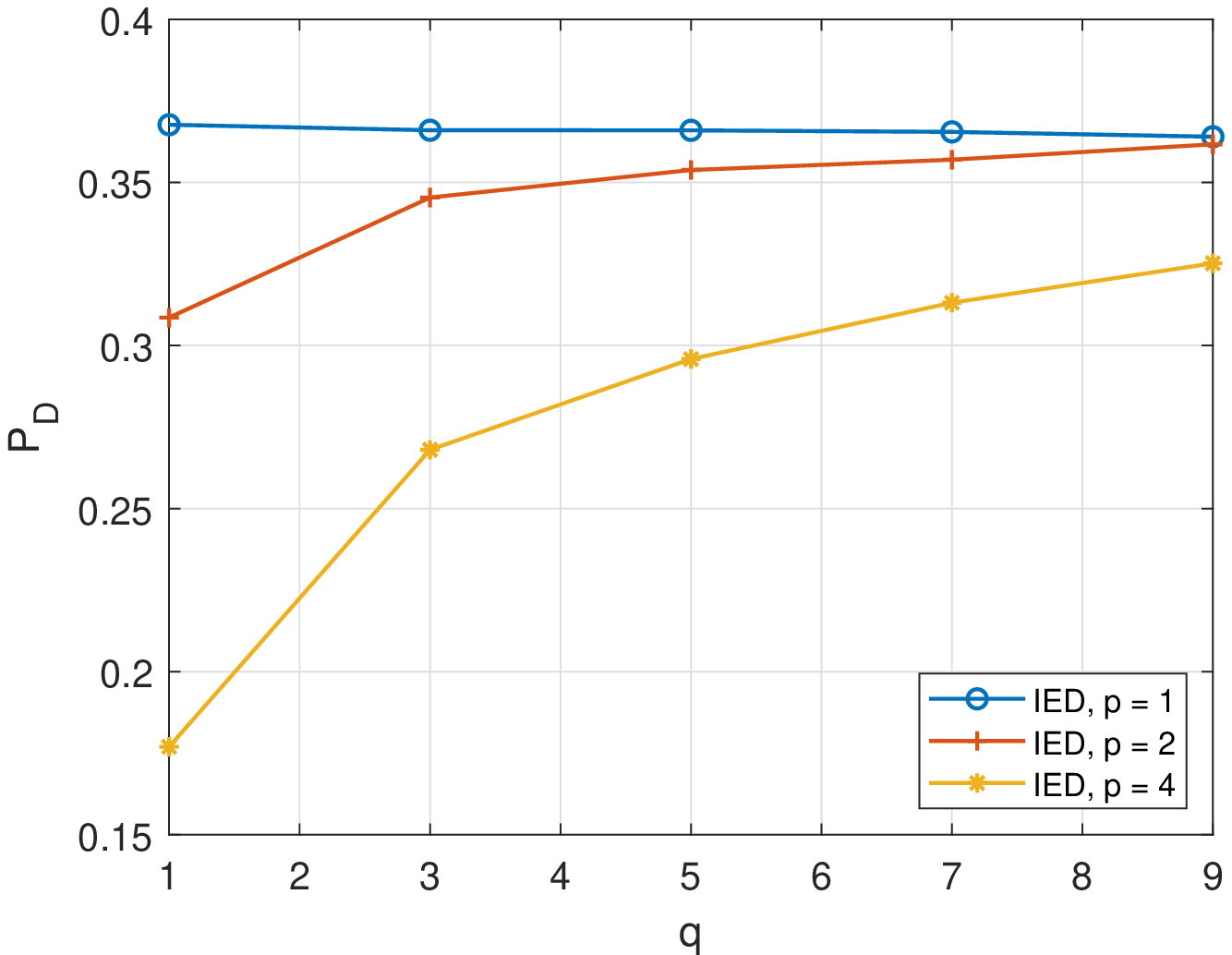}
\caption{Probability of detection vs. $q$ for IED under different values of $p$.}
\label{Pd_q}
\end{subfigure}
\vspace{\floatsep}
\begin{subfigure}[b]{0.5\textwidth}
\centering
\includegraphics[width=3.15in]{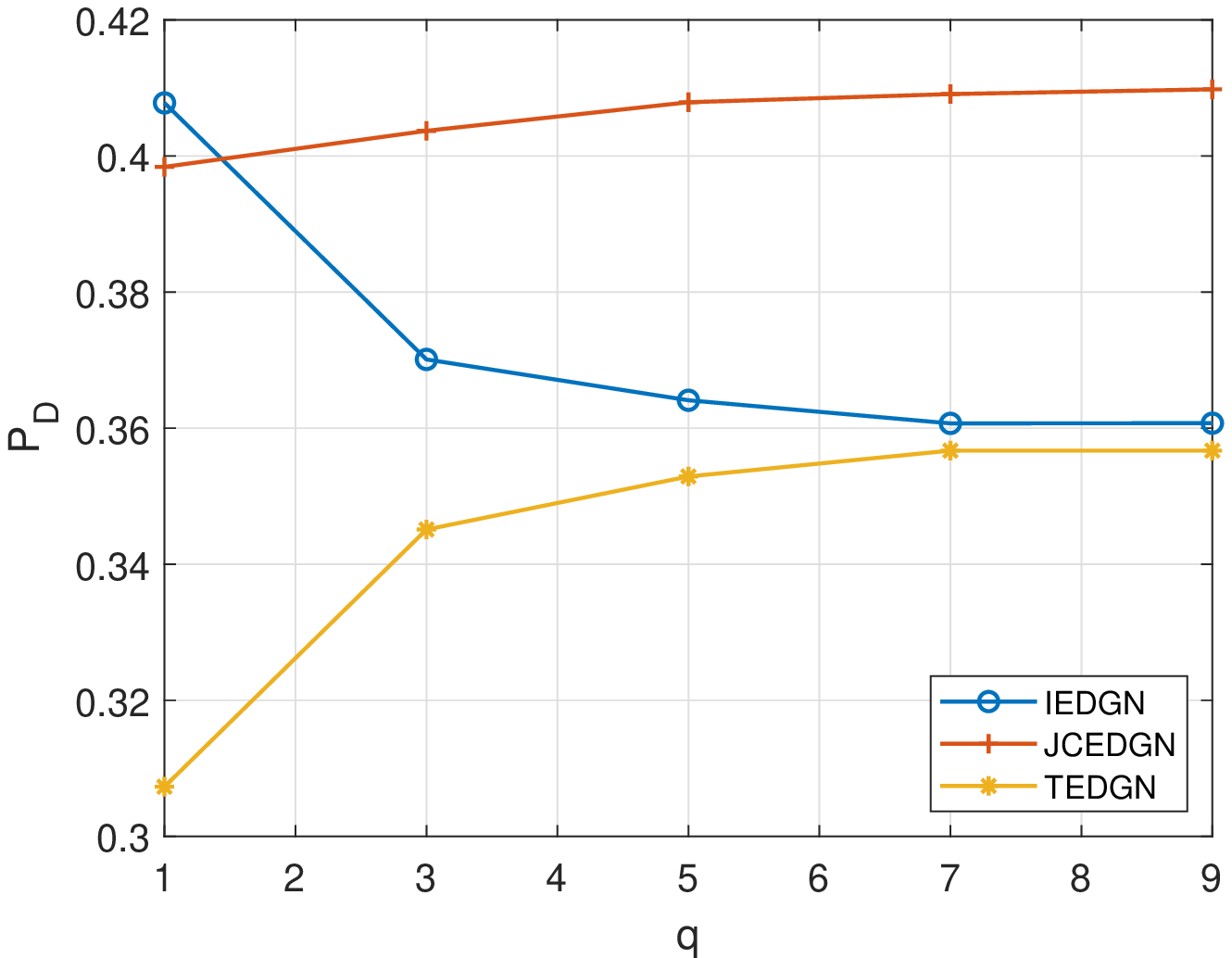}
\caption{Probability of detection vs. $q$ for different detectors.}
\label{Pd_q_1}
\end{subfigure}
\caption{Detection performance for various detectors' parameters and detectors when $P_s=15$ dBm and  $M=1$. \vspace{-5mm}} 
\label{Pd_q_2}
\end{figure}

Fig. \ref{AUC_fig} illustrates the AUC values obtained through analytical and numerical calculations for different detectors across various transmit power values. The IED surpasses the TED in maximizing the AUC across different $P_s$ values. Furthermore, the figure demonstrates that the IEDGN outperforms the IED under Laplacian noise. This indicates that the choice of the detector depends on the noise type and $P_s$ value, with the proposed IED exhibiting superior overall performance compared to the TED.

\begin{figure}
\centering
\begin{subfigure}[b]{0.5\textwidth}
\centering
\includegraphics[width=3.15in]{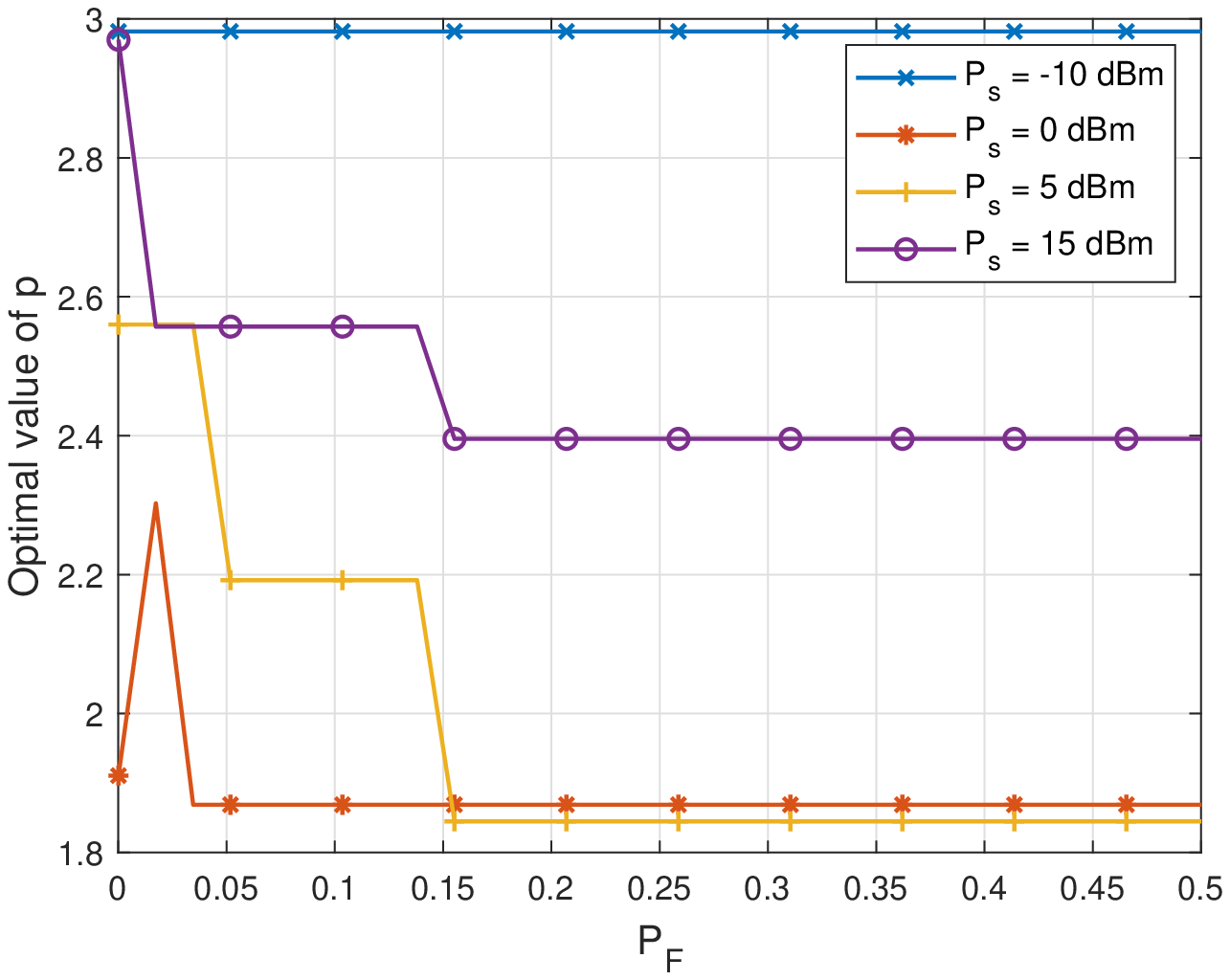}
\caption{$N=256$.}
\label{p_opt_1}
\end{subfigure}
\vspace{\floatsep}
\begin{subfigure}[b]{0.5\textwidth}
\centering
\includegraphics[width=3.15in]{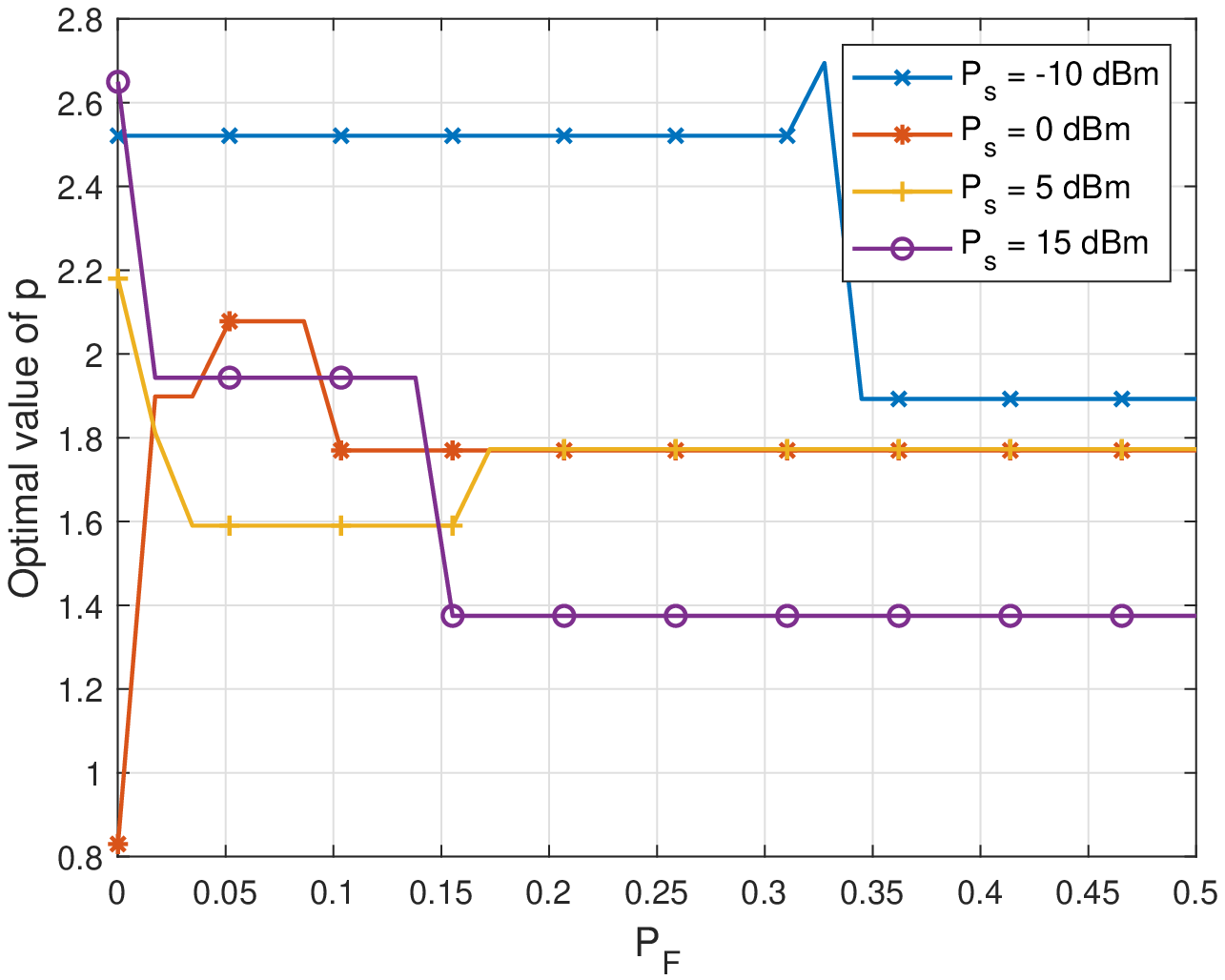}
\caption{$N=512$.}
\label{p_opt_2}
\end{subfigure}
\caption{The optimum value of $p$ vs. probability of false alarm for different fixed values of $P_s$ and $N$  when $M=1$. \vspace{-5mm}} \label{p_opt}
\end{figure}

Fig. \ref{Pd_q} illustrates that a reduction in the parameter $q$, which signifies a greater presence of impulsive noise, enhances detection efficiency at lower values of $p$. This enhancement is particularly pronounced as $q \rightarrow 0^+$, a condition indicative of highly non-Gaussian noise environments. Conversely, as $q$ increases, moving towards Gaussian noise characteristics, the advantageous impact of a smaller $p$ diminishes. Furthermore, Fig. \ref{Pd_q_1} presents the relationship between the probability of detection and the parameter $q$. The IEDGN demonstrates a high starting point at $q = 1$ but declines by the time $q$ reaches $3$. This pattern suggests an initial robustness against noise, gradually decreasing as the noise becomes more non-Gaussian. In contrast, the TEDGN initiates at a lower point but shows a stabilization trend as $q$ increases. Meanwhile, the JCEDGN remains relatively constant across different $q$ values, implying a resilience to variations in noise type, though it may not be as effective in purely Gaussian noise environments.

In Fig. \ref{p_opt}, the optimal value of $p$ is depicted as a function of $P_s$ and $N$ values. The calculation involves $1000$ evenly spaced points between $0.1$ and $3$. The results reveal that with an increasing $P_F$, the optimal $p$ value decreases at a nearly constant rate for high $P_F$ values. However, when $P_F$ is low, the optimal $p$ value decreases rapidly. This suggests that TED may not always correspond to the optimal $p$ value. It emphasizes the significance of selecting the appropriate $p$ value to maximize detection accuracy and minimize $P_F$. 

\begin{figure*}[t]
\begin{eqnarray}\label{con_0}
\text{Cov}(Z_1, Z_2) && = \mathbb{E}\left[\left( Z_1-\mathbb{E}(Z_1)\right)\left (Z_2-\mathbb{E}(Z_2)\right)^*\right]\\\nonumber
&&=\mathbb{E}\left[\left( \sum^{N-1}_{n=0}|y(n)|^2-\mathbb{E}\left( \sum^{N-1}_{n=0}|y(n)|^2\right)  \right) \left( \sum^{N-2}_{m=0}y^*(m+1)y(m)-\mathbb{E}\left( \sum^{N-2}_{m=0}y^*(m+1)y(m)\right)  \right)\right]\\\nonumber
&&=\mathbb{E}\left[\left( \sum^{N-1}_{n=0}\left[|y(n)|^2-\mathbb{E}\left( |y(n)|^2\right)  \right]  \right) \left( \sum^{N-2}_{m=0} \left[y^*(m+1)y(m)-\mathbb{E}\left( y^*(m+1)y(m)\right) \right]  \right)\right]\\\nonumber
&&=\sum^{N-1}_{n=0}\sum^{N-2}_{m=0}\mathbb{E}\left[\left( |y(n)|^2-\mathbb{E}\left( |y(n)|^2\right)   \right) \left(  y^*(m+1)y(m)-\mathbb{E}\left( y^*(m+1)y(m)\right)   \right)\right].
\end{eqnarray}	
\hrulefill
\end{figure*}

\begin{figure}[t]
\centering
\includegraphics[width=3.15in]{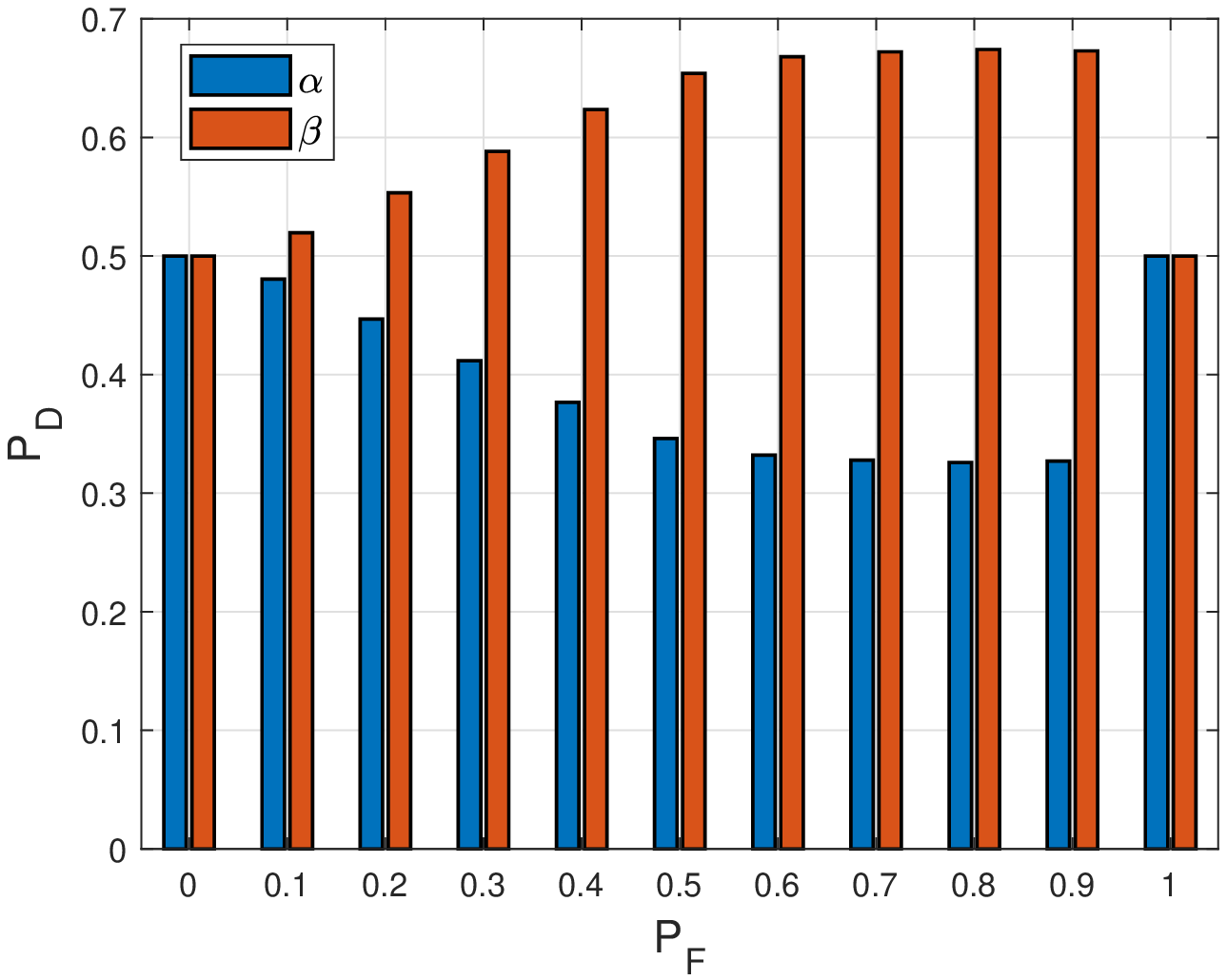}
\caption{The suboptimal value of $\alpha$ and $\beta$ vs. probability of false alarm for $P_s=10$ dBm  when $M=1$. }
\label{Opt_weights_JSEC}
\end{figure}

\begin{figure}
\centering
\begin{subfigure}[b]{0.5\textwidth}
\centering			\includegraphics[width=3.15in]{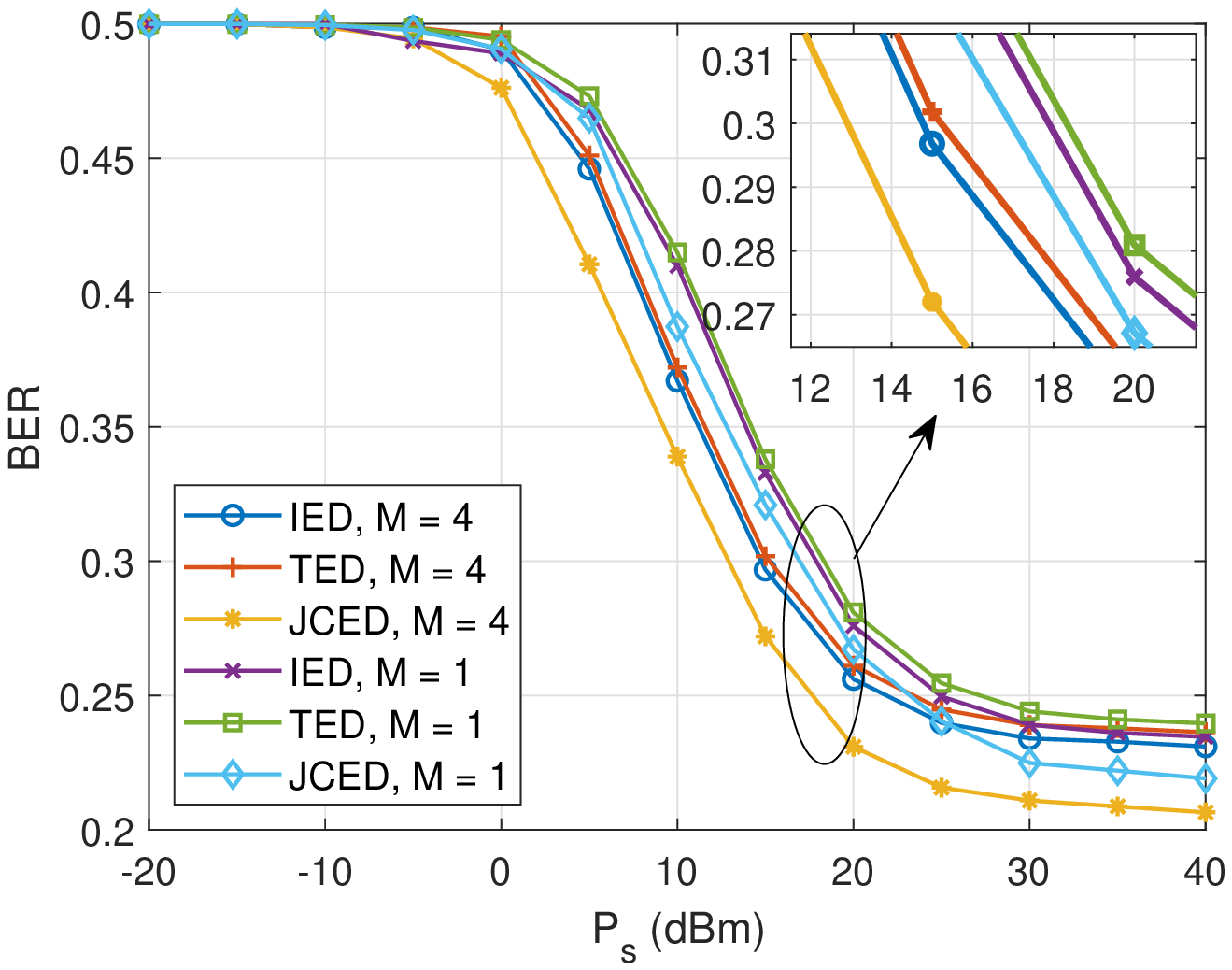}
\caption{BER vs. $P_s$ for different detectors without DIC technique.}
\label{BER_multi_fig1}
\end{subfigure}
\vspace{\floatsep}
\begin{subfigure}[b]{0.5\textwidth}
\centering
\includegraphics[width=3.15in]{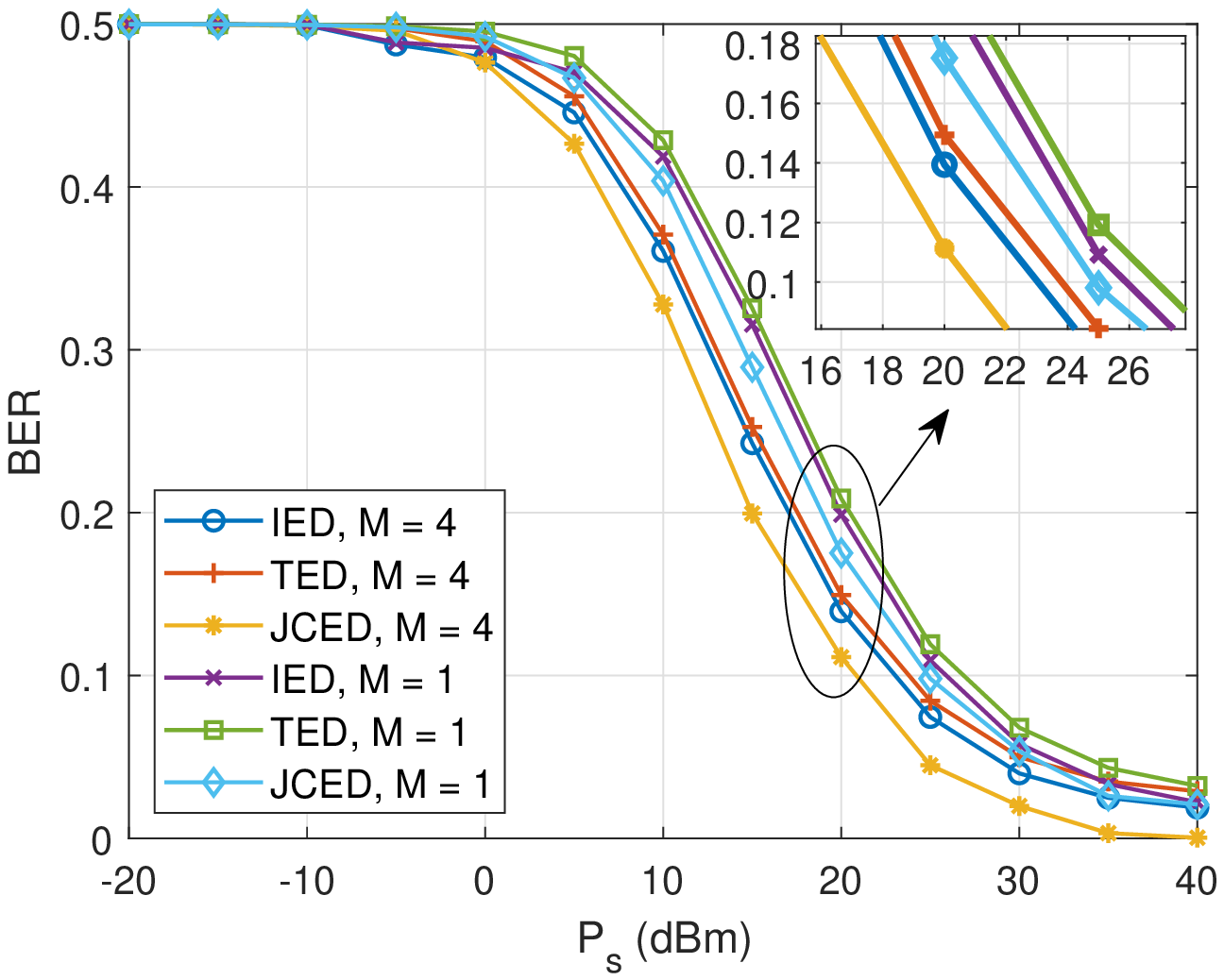}
\caption{BER vs. $P_s$ for different detectors with DIC technique.}
\label{BER_multi_fig2}
\end{subfigure}
\caption{Detection performance for different detectors when $M=4$. \vspace{-5mm}} 
\label{BER_multi_fig}
\end{figure}

\begin{figure}[t]
\centering
\includegraphics[width=3.15in]{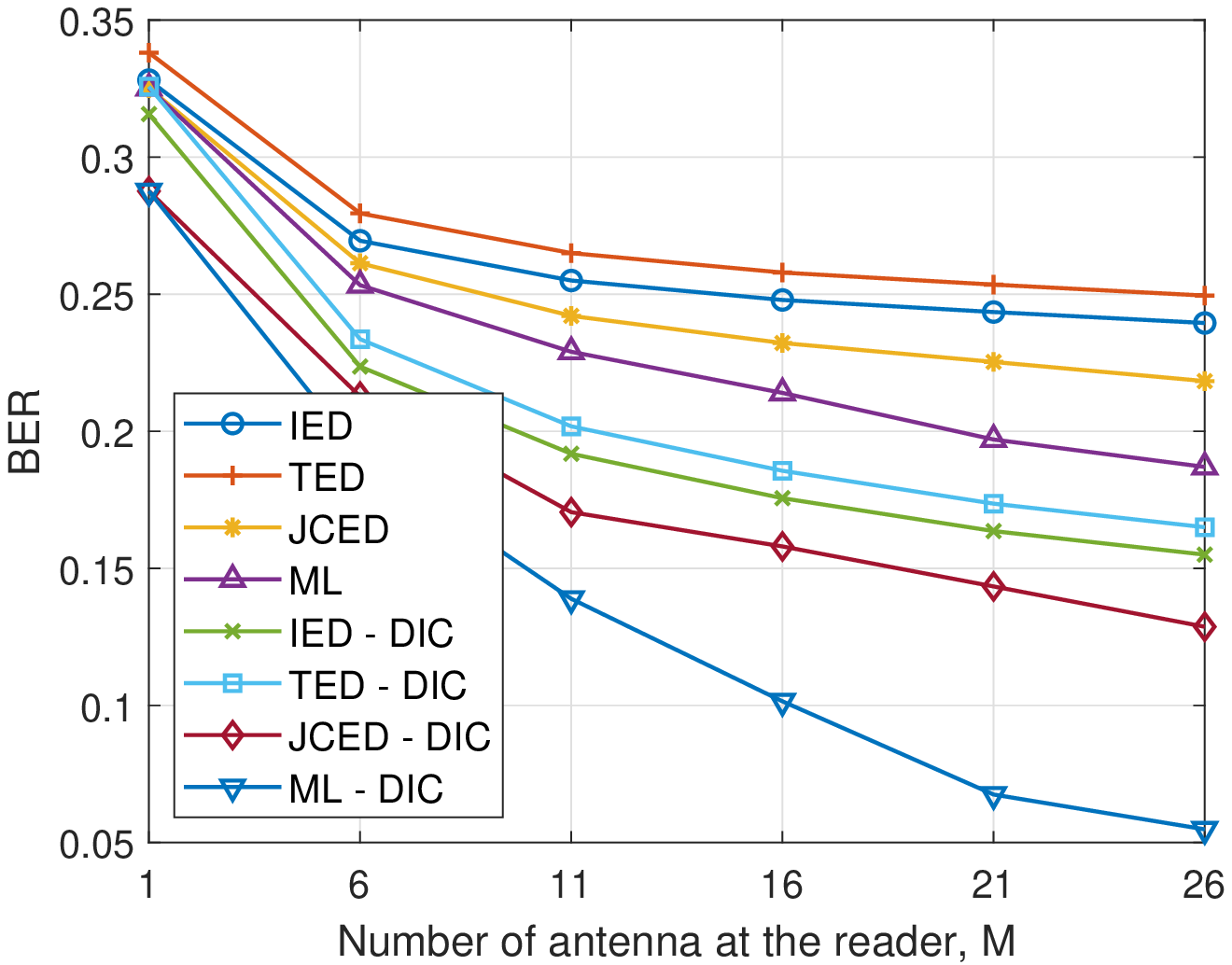}
\caption{BER vs. the number of reader antennas for different detectors for  $P_s=15$ dBm. } 
\label{antennafig}
\end{figure}

In Fig. \ref{Opt_weights_JSEC}, the suboptimal values of $\alpha$ and $\beta$ are plotted, maximizing $P_D$ vs. $P_F$ for the proposed JCED. The results provide insights into the role of $\alpha$ and $\beta$ in optimizing JCED's detection performance. For a low value of $P_F$, the suboptimal values suggest nearly equal weights assigned to both energy and correlation. Hence, the energy and correlation features contribute almost equally to lower false alarm rates. However, as $P_F$ increases, the value of $\beta$ becomes higher than $\alpha$, implying a higher weight for correlation in improving detection performance. 

In Fig. \ref{BER_multi_fig1}, we examine the efficacy of the proposed detectors as a function of transmit power.  The reader is equipped with $M=1$ and $M=4$ antennas. For $M=1$, the JCED outperforms IED and TED with a lower BER as $P_s$ increases. With $M=4$ antennas, all detectors achieve a lower BER across the $P_s$ regions, demonstrating the significant benefits of antenna diversity. The JCED consistently outshines both IED and TED, with the performance discrepancy becoming increasingly pronounced at higher $P_s$ values. For instance, with $P_s=15$ dBm and $M=4$, the JCED attains a BER gain of $28.68\%$. Correspondingly, the JCED exhibits a $14.71\%$ more favorable BER performance than the TED. Fig. \ref{BER_multi_fig2} presents the BER outcomes when employing the DIC technique. With just a single antenna, the BER reduction is quite substantial, emphasizing the pivotal role of DIC in curtailing interference. As the number of antennas increases to $M=4$, DIC leads to an even more marked BER reduction, especially at higher  $P_s$ levels. For example, when power $P_s$ is set to $15$ dBm and $M= 4$, the JCED achieves a BER improvement of $48.2\%$. In comparison, with the same settings, JCED shows a $27.9\%$ better BER performance than TED.

Fig. \ref{antennafig} displays the BER performance as a function of the number of antennas, $M$, at the reader, at a fixed $P_s$ of $15$ dBm. The descending BER  trend is evident as the number of antennas increases, indicating the positive impact of antenna diversity on signal detection accuracy. Both IED and TED show a gradual decrease in BER with more antennas, while JCED demonstrates a more pronounced reduction, signifying its superior performance in multi-antenna systems. Incorporating DIC enhances the performance of all detectors, with the most significant improvement observed in the JCED-DIC combination, which consistently outperforms the others across the range of antenna counts. For instance, without  DIC, the JCED exhibits a performance gain of $9.2\%$ compared to the IED. This improvement is observed for a range of antenna numbers from $M=6$ to $M=11$. When DIC is applied, the  JCED gain further escalates to $24.7\%$ compared to the IED for the same range of antennas.  Conversely, without  DIC, the IED records a gain of $5.9\%$ for antenna numbers extending from $M=6$ to $M=11$. However, upon the integration of DIC, the IED method demonstrates an enhanced gain, reaching $16.8\%$.

\section{Conclusion}
The performance of AmBC systems relies on the reader's ability to detect signals, which must recover data amidst interference and weak signals. To address these challenges, our study developed two innovative signal detectors that offer higher detection accuracy: (a) JCED and (b) IED. The first capitalizes on correlations between adjacent signal samples to improve the detection of weak signals. In contrast, the second utilizes non-linear signal combining to enhance accuracy, which is particularly effective in non-Gaussian noise environments.

Both proposed detectors show superior performance compared to the classical TED, especially under a multi-antenna reader setting. This opens new avenues for future studies to investigate the application of non-linear signal processing and to integrate machine learning techniques further to improve detection performance. Moreover, the concept of JCED can be extended to multiple-input and multiple-output (MIMO) systems.

\appendix

\subsection{ Proof of Proposition \ref{propro_2}}\label{Pro_1}

The correlation between two random variables $Z_1$ and $Z_2$ is defined in \eqref{con_0}. As a result, the last equality can be decomposed as given in \eqref{cov_2}.
\begin{figure*}[t]
\begin{eqnarray}\label{cov_2}
\hspace{-120mm} \text{Cov}(Z_1, Z_2)&& = \sum^{N-1}_{n=0}\sum_{\substack{m=0\\m\neq n\\m\neq n-1}}^{N-2}\mathbb{E}\left( |y(n)|^2-\mathbb{E}\left( |y(n)|^2\right)   \right)\underbrace{\mathbb{E} \left(  y^*(m+1)y(m)-\mathbb{E}\left( y^*(m+1)y(m)\right)   \right)}_{=\;0}\nonumber\\
&&+\sum_{\substack{n=0\\m= n}}^{N-2}\mathbb{E}\left( \left( |y(n)|^2-\mathbb{E}\left( |y(n)|^2\right)   \right)  \left(  y^*(n+1)y(n)-\mathbb{E}\left( y^*(n+1)y(n)\right)   \right) \right)\nonumber\\
&&+\sum_{\substack{n=1\\m= n-1}}^{N-1}\mathbb{E} \left(  \left( |y(n)|^2-\mathbb{E}\left( |y(n)|^2\right)   \right)  \left(  y^*(n)y(n-1)-\mathbb{E}\left( y^*(n)y(n-1)\right)\right)   \right) .
\end{eqnarray}	
\hrulefill
\end{figure*}
By substituting \eqref{received_signal_2} in \eqref{cov_2} under different hypotheses, we obtain
\begin{align}\label{cov_3}
&\text{Cov}(Z_1, Z_2)  = \sum^{N-2}_{n=0}\mathbb{E}\left( |h_js(n)+w(n)|^2-\left( |h_js(n)|^2+\sigma_w^2\right)   \right)\nonumber \\
&\left( w(n)h^*_js^*(n+1)+w(n+1)h^*_js^*(n)+w(n)w^*(n+1)  \right)\nonumber \\
&+\sum^{N-1}_{n=1}\mathbb{E}\left( |h_js(n)+w(n)|^2-\left( |h_js(n)|^2+\sigma_w^2\right)   \right)  \nonumber \\
&\left( w(n-1)h^*_js^*(n)+w^*(n)h_js(n-1)+w^*(n)w(n)  \right),
\end{align}
where $j \in\{0, 1\}$. Using the fact that for $w(n)\sim \mathcal{CN}(0,\sigma_w^2)$, we have $ 
\mathbb{E}\left(w^2(n)\right)=\mathbb{E}\left({w^*}^2(n)\right)=\mathbb{E}\left(|w(n)|^2w(n)\right)=\mathbb{E}\left(|w(n)|^2w^*(n)\right)=0$, \eqref{cov_3} can be further simplified into
\begin{align}
&	\text{Cov}(Z_1, Z_2)= \sum^{N-2}_{n=0}\mathbb{E}\left(|w(n)|^2|h_j|^2s(n)s^*(n+1)  \right)\nonumber\\&+\sum^{N-1}_{n=1}\mathbb{E}\left(|w(n)|^2|h_j|^2s^*(n)s(n+1)  \right)\nonumber\\
&=2\sigma_w^2|h_j|^2\sum^{N-2}_{n=0}s(n)s^*(n+1) =2\sigma_w^2|h_j|^2 R_{ss}(1).
\end{align}
Thus, the proof is complete.

\subsection{ Proof of Proposition \ref{propro_1}}\label{app_2}
To begin with, we decompose $w$ into a product of the squared roots of independent  Gamma and a CSC Gaussian random variables,  i.e., $w=\sqrt{G}C$, where  
\begin{equation}\label{pdf_g}
f_G(g)=\frac{q^q}{\Gamma(q)}g^{q-1}e^{-qg},~g\in \mathbb{R}^{+},
\end{equation}
and $C\sim\mathcal{CN}(0,\sigma_w^2)$. Accordingly, the received signal at the reader under different hypotheses, $j \in\{0, 1\}$, can be decomposed as follows:
\begin{equation}\label{yyyy}
y=h_js(n)+w(n)=  \underbrace{h_js(n)}_{x_1}  + \underbrace{w(n)}_{x_2}, 
\end{equation}
where $\Pr(x_1|{h_j,s(n)})\sim\mathcal{CN} (0,~|h_j|^2P_s)$ and $\Pr(x_2|{G})\sim \mathcal{CN} (0, G\sigma_w^2)$
using the fact that $\{h_j, s(n), G, X\}$ are all mutually independent random variables. As can be observed, \eqref{yyyy} is the sum of two Gaussian random variables. By introducing the auxiliary variable $v:= x_1+ x_2$, we have the following conditional probability:
\begin{align}
\Pr\left(v\big|h_j, s(n), G\right)& \sim 
\mathcal{CN}(0,|h_j|^2P_s+G\sigma_w^2) \nonumber\\&\overset{a}{=}\left(|h_j|^2P_s+G\sigma_w^2\right)^{\frac{1}{2}}  \underbrace{ \mathcal{CN}(0,1)}_{x_3}.
\end{align}
For step (a), the variance of the Gaussian random variable is extracted to obtain the probability distribution using the moment-generating function. Accordingly, the absolute moments of $v$ can be expressed as
\begin{equation}\label{absolut_value}
\mathbb{E}\left[|v|^{p}\right]=\mathbb{E}\left[\left(|h_j|^2P_s+G\sigma_w^2\right)^{\frac{p}{2}}\right]\mathbb{E}\left[|x_3|^{p}\right].
\end{equation}
Consequently, $|x_3|$ is Rayleigh distributed with PDF  \cite[Eq. (2.17)]{simon2002probability}: $f_{|X_3|}(x_3)=2x_3\exp\left(-x_3^{2}\right),~x_3\geq 0$. Therefore, the $p$-th order moment of $|x_3|$ is given by \cite[Eq. (2.15.5.4)]{prudnikov2003integrals}: $\mathbb{E}\left[|x_3|^{p}\right] = \Gamma\left(\frac{p}{2}+1\right)$. To calculate the remaining factor of \eqref{absolut_value}, we average out the Gamma distributed parameter $G$. The first step is to perform the following transformation \cite[Eq. (8.4.2.5)]{prudnikov2003integrals}:
\begin{align}\label{ccccc}
\left(|h_j|^2P_s+G\sigma_w^2\right)^{\frac{p}{2}}&=\frac{\left({|h_j|^2P_s}\right)^{\frac{p}{2}}}{\Gamma(-\frac{N}{2})}\MeijerG{1,1}{1,1}{\frac{p}{2}+1 \\ 0 } { \frac{G}{\gamma_j} },
\end{align}
where it can be applied to any number $p$ other than positive even numbers greater than or equal to $2$. Since when $p = \{2, 4, 6,  \ldots\}$, \eqref{ccccc} always returns singularity in its denominator. Then, using \eqref{pdf_g}, \eqref{ccccc} and exploiting \cite[Eq. (2.24.3.1)]{prudnikov2003integrals}, it yields
\begin{align}\label{varaince_genralized_3}
\mathbb{E}\left[\left|h_js(n)+w(n)\right|^{p}\Big|{\mathcal{H}_{j}}\right]&=\frac{\left({|h_j|^2P_s}\right)^{\frac{p}{2}}}{\Gamma(q)\Gamma(-\frac{p}{2})} \nonumber\\& \times \MeijerG{1,2}{2,1}{ 1-q,~\frac{p}{2}+1 \\ 0 } { \frac{1}{\gamma_j q} }.
\end{align}
Thus, the proof is completed. For the alternative case where $p\geq 2$ is an even integer, \eqref{varaince_genralized_3} can be relaxed to  \eqref{varaince_genralized_2} using the binomial expansion.

\subsection{ Proof of Proposition \ref{propro_3}}\label{Pro_3}
Let us simplify the integral in \eqref{AUC_gg} by introducing a change of variable as follows:
\begin{align}
\hat{\lambda}&=\frac{\lambda-\mathbb{E}\left\{W|\mathcal{H}_0\right\}}{\sqrt{{\text{Var}\left\{W|\mathcal{H}_0\right\}}}}, \quad a=\frac{\mathbb{E}\left\{W|\mathcal{H}_0\right\}-\mathbb{E}\left\{W|\mathcal{H}_1\right\}}{\sqrt{{\text{Var}\left\{W|\mathcal{H}_1\right\}}}},\nonumber\\ b&=\sqrt{\frac{\text{Var}\left\{W|\mathcal{H}_0\right\}}{\text{Var}\left\{W|\mathcal{H}_1\right\}}}.   
\end{align}
In this way, \eqref{AUC_gg} can be rewritten as follows:
\begin{equation}\label{AUC_1}
A(\gamma_j) = \frac{1}{2\pi}\int^{\infty}_{-\infty}\mathcal{Q}(\hat{\lambda} b +a)\text{exp}\left(\frac{-\hat{\lambda}^2}{2}\right)\text{d}\hat{\lambda}.
\end{equation}
Based on the definition of $\mathcal{Q}$ function, \eqref{AUC_1} can be expressed as
\begin{equation}\label{AUC_2}
A(\gamma_j) = \frac{1}{2\pi}\int^{\infty}_{-\infty}\int^{\infty}_{\hat{\lambda} b +a} \text{exp}\left(-\frac{t^2}{2}-\frac{\hat{\lambda}^2}{2}\right)\text{d}t\;\text{d}\hat{\lambda}.
\end{equation}
The exponential function is doubly integrated over the $(t,\hat{\lambda})$ plane. As a result,  a rotation of the coordinates will simplify the integral evaluation.  Let  $\theta$ refer to the angle between $t = \hat{\lambda} b +a$ and the positive $x$-axis. The  following transformations generate a new coordinate system: $\hat{\lambda} =f(\bar{\lambda},\hat{t})= \bar{\lambda} \text{cos}\theta-\hat{t}\text{sin}\theta$ and $t =g(\bar{\lambda},\hat{t})= \bar{\lambda}$. As a result of the new coordinate system, the integral in \eqref{AUC_2} can be rewritten as follows:
\begin{equation}
A(\gamma_j)=\frac{1}{2\pi}\!\int^{\infty}_{-\infty}\int^{\infty}_{\frac{a}{\sqrt{b^2+1}}} \!\!\!\!\text{exp}\left(\!-\!\frac{f(\bar{\lambda},\hat{t})^2}{2}\!-\!\frac{g(\bar{\lambda},\hat{t})^2}{2}\right)|J|\;\text{d}\hat{t}\;\text{d}\bar{\lambda},
\end{equation}
where $|J|$ is the Jacobian, it equals one. With the integration order reversed, the square on $\bar{\lambda}$ in exponent completed, and applying Gaussian PDF properties, we obtain
\begin{align}
& A(\gamma_j)=\frac{1}{2\pi}  \int^{\infty}_{\frac{a}{\sqrt{b^2+1}}} \text{exp}\left( -\frac{\hat{t}^2}{2}\right) \text{d}\hat{t}= \mathcal{Q}\left( \frac{a}{\sqrt{b^2+1}} \right),
\end{align}
which completes the proof.

\bibliographystyle{ieeetr}
\bibliography{IEEEabrv,ref}
\newpage

\end{document}